\documentclass{amsart}

\usepackage{amssymb, amsmath, amsthm}
\usepackage{mathrsfs}

\usepackage{mathabx}
\usepackage[active]{srcltx}
\usepackage{verbatim}
\usepackage{tikz}
\usepackage{tikz-cd}
\usetikzlibrary{cd}
\usepackage{physics}
\usepackage[shortlabels]{enumitem}

\usepackage[colorlinks,linkcolor={blue},citecolor={blue},urlcolor={red},]{hyperref}

\usepackage[fleqn,tbtags]{mathtools}

\usepackage{newsymbol}

\let\mathcal \undefined
\def\mathcal{\mathscr}

\let\emptyset \undefined
\let\ge    \undefined
\let\le    \undefined
\newsymbol\le     1336 
\newsymbol\ge     133E \let\geq\ge
\newsymbol\emptyset  203F
\newsymbol\notle    230A
\newsymbol\notge    230B

\theoremstyle{plain}
\newtheorem{theorem}{Theorem}[section]
\theoremstyle{remark}
\newtheorem{remark}[theorem]{Remark}
\newtheorem{example}[theorem]{Example}
\newtheorem{definition}[theorem]{Definition}
\theoremstyle{plain}
\newtheorem{corollary}[theorem]{Corollary}
\newtheorem{lemma}[theorem]{Lemma}
\newtheorem{proposition}[theorem]{Proposition}

\numberwithin{equation}{section}

\def\N{{\mathbb N}}
\def\Z{{\mathbb Z}}

\def\R{{\mathbb R}}
\def\C{{\mathbb C}}

\newcommand{\eps}{\varepsilon}
\newcommand{\la}{\lambda}

\newcommand{\om}{\omega}
\newcommand{\Om}{\Omega}

\renewcommand{\P}{{\mathbb P}}
\renewcommand{\S}{{\mathbb S}}
\newcommand{\T}{{\mathbb T}}

\newcommand{\calA}{\mathscr{A}}
\newcommand{\calB}{\mathscr{B}}
\newcommand{\calF}{\mathscr{F}}

\newcommand{\calK}{\mathscr{K}}
\newcommand{\calL}{\mathscr{L}}
\newcommand{\calM}{\mathscr{M}}

\newcommand{\calP}{\mathscr{P}}
\newcommand{\calS}{\mathscr{S}}
\newcommand{\calT}{\mathscr{T}}

\newcommand{\Ker}{{\mathsf{N}}}
\newcommand{\Ran}{{\mathsf{R}}}

\newcommand{\n}{\Vert}

\newcommand{\iprod}[2]{( #1|#2 )}

\newcommand{\one}{{{\bf 1}}}

\newcommand{\wh}{\widehat}

\newcommand{\ud}{\,{\rm d}}

\newcommand{\Proj}{{\rm Proj}}
\newcommand{\ProjH}{{\rm Proj}(H)}

\newcommand{\SH}{\calS(H)}
\newcommand{\SHP}{\SH/\!{}_P}
\newcommand{\SM}{\calS(\calM)}
\newcommand{\simM}{\sim_{\!\calM}}
\newcommand{\simP}{\sim_{_P}}
\newcommand{\SHM}{\SH/\!{}_{\calM}}
\newcommand{\ESH}{\calP(H)}
\newcommand{\ESHP}{\ESH/\!{}_P}
\newcommand{\ESHM}{\ESH/\!{}_\calM}

\newcommand{\tauM}{\tau}
\newcommand{\tauMt}{\tau_{/\!\calM}}

\allowdisplaybreaks

\begin{document}

\title[Indiscernibility]{Indiscernibility of quantum states}

\author{Jan van Neerven \& Marijn Waaijer}

\address{Delft Institute of Applied Mathematics\\
Delft University of Technology\\P.O. Box 5031\\2600 GA Delft\\The Netherlands}

\email{j.m.a.m.vanneerven@tudelft.nl, waaijermarijn@gmail.com}

\date{\today}


\begin{abstract}
 This paper provides a systematic study of the operational idea that a quantum ``state'' is only defined up to what can be distinguished by a chosen family of observables. Concretely, any von Neumann algebra of observables $\calM$ induces an equivalence relation on pure and mixed states by declaring two preparations indiscernible when they give identical statistics for every observable in $\calM$. The corresponding quotient, the \emph{Holevo space} associated with $\calM$, is the effective (relational) state space of the experiment, explicitly dependent on the observer’s available measurements.

We analyse the resulting geometry and topology of these quotients, and prove a context-complete classical representation theorem: for every von Neumann algebra $\calM$ there is a canonical lift $a\mapsto \wh a$ to bounded continuous functions on the Holevo space, reproducing expectation values pointwise. In the commutative case this reduces to ordinary probability theory on the joint spectrum.

The framework is illustrated in explicit examples, including position measurements of a free particle and polarisation measurements in the qubit, EPR, and Bell settings. In particular, in the EPR scenario Charlie’s joint observable defines a simplex of joint outcome distributions, whereas the Alice/Bob marginal viewpoint collapses the effective description to a lower-dimensional space by ``forgetting'' the correlation parameter. We show that by varying the polariser settings, the indiscernibility classes become conjugated (and generically reshuffled), and different settings are typically incompatible at the level of observable algebras.
\end{abstract}

\maketitle

\section{Introduction}

In this paper, we present a systematic study of the operational idea that a ``state'' is only defined up to what can be distinguished by the measurements one allows. Concretely, any given set (or algebra) of observables $\calM$ induces an equivalence relation on quantum states: for an observer with access to $\calM$, two preparations are operationally the ``same'' if they give identical statistics for every measurement built from $\calM$ (Definition~\ref{def:indiscernible}).
The corresponding quotient of the (pure or mixed) state space is the {\em Holevo space} associated with $\calM$ (Section~\ref{sec:Holevo}); it is the effective state space of the experiment and depends on the observational standpoint encoded by $\calM$.

In our setup, an observer's accessible information is encoded by the abelian von Neumann algebra $\calM$ generated by commuting observables representing what that observer can jointly measure.
For every measurement context $\calM$ there is a canonical ``classicalisation'' which assigns to each quantum observable $a\in\calM$ a bounded continuous function $\wh a$ on an explicit Hausdorff space of equivalence classes $[h]$, reproducing expectation values pointwise,
$$\wh a([h])=\iprod{ah}{h}$$ (Theorem~\ref{thm:hidden}).
This map becomes particularly insightful in multi-observer scenarios.
Different observers typically have access to different algebras of observables, and hence to different notions of ``sameness''.
In this sense, the Holevo space is a genuinely relational state space: it does not encode an observer-independent description, but exactly the empirical content of the system relative to the specified measurement context.

A concrete illustration is provided by the EPR setting (Section \ref{subsec:C-EPR}). Defining the $n$-simplex $$\Delta^n := \Bigl\{x\in\R_+^{n+1}:\ \sum_{j=1}^{n+1}x_j = 1\Bigr\},$$ the Holevo space for Charlie’s joint observable recording outcome pairs can be represented as the three-parameter Holevo space $\Delta^3$. In contrast, the Holevo space for the combined marginal perspectives of Alice and Bob gives the two-parameter space
$\Delta^1\times \Delta^1$ (Section \ref{subsec:C-EPR}, \ref{subsec:AB-EPR}). This is a natural result, given the fact that $\Delta^1$ is the Holevo space for a polarisation measurement on a single photon (Section \ref{subsec:Holevo-qubit}). Geometrically, passing from Charlie to the Alice/Bob marginals amounts to forgetting precisely the correlation parameter.

The Bell experiment highlights how changing the measurement context reshuffles the relational state space. For each fixed polariser setting $(\gamma_A,\gamma_B)$, Charlie's Holevo space is again $\Delta^3$ and the equivalence relation is obtained by conjugation with the local rotation $R^{\gamma_A}\otimes R^{\gamma_B}$ (Theorem~\ref{thm:hide-2-intrinsic}).
However, the underlying equivalence classes depend strongly on the setting:
different settings are generically incompatible at the level of observable algebras (Theorem~\ref{thm:incompatibility}).
These results visualise the contextuality of the Bell scenario geometrically in terms of how ``effective state spaces'' vary with the measurement perspective.

The map $\wh a$ sends a state $\varrho$ to the list of numbers
$\tr(a\varrho)$ with $a\in\calM$, and thus
encodes the ``access through commutative algebras'' in the sense of the Bohrification programme \cite[Introduction]{Landsman}.
This programme takes at its point of departure that, although a quantum algebra of observables is typically noncommutative, any concrete observational access to it comes through a commutative subalgebra (a ``classical context'') corresponding to a given choice of compatible measurements.
In this light, the incompatibility result of Theorem~\ref{thm:incompatibility}
provides a concrete, experiment-level motivation for Bohrification: there simply is no single commutative algebra that contains all the relevant contexts -- ``classical descriptions'' are necessarily contextual.

\smallskip
The paper is organised as follows. After a brief discussion of observables and states in Section \ref{sec:observales-states}, in Section \ref{sec:indiscernible} we introduce the equivalence relation of indiscernibility and characterise its classes in terms of commutants and unitaries.  In Section \ref{sec:Holevo} we define the Holevo space of equivalence classes.  The ``classicalisation'' construction is carried out in Section \ref{sec:hidden}.  Sections \ref{sec:examples}, \ref{sec:free-particle}, \ref{sec:EPR_and_Bell} illustrate the theory in key examples, including finite commuting families and position observables for a free particle, the qubit, and the EPR and Bell experiments.

\medskip\noindent
{\em Notation and conventions.}
Throughout this paper, $H$ is a nonempty (possibly finite-dimensional) separable complex Hilbert space with inner product $\iprod{\,\cdot\,}{\,\cdot\,}$ linear in the first argument and sesquilinear in the second.  We write $\calL(H)$ for the space of bounded operators on $H$, $\Proj(H)$ for the lattice of its orthogonal projections, and $\mathscr S(H)$ and $\mathscr{P}(H)$ for the convex set of mixed states and its subset of pure states, respectively. We use the conventions $\N = \{0,1,2,\dots\}$ and $\R_+ = [0,\infty)$.

\section{Observables and states}\label{sec:observales-states}

In the mathematical formulation of quantum mechanics, states and observables are defined as follows.

\begin{definition}[States and observables]
Let $H$ be a separable complex Hilbert space.

\begin{itemize}
 \item A {\em state} is a positive trace class operator on $H$ with unit trace.
 \item A {\em pure state} is an extreme point of the convex set of all states.
 \item An {\em observable} is a projection-valued measure with values in $\Proj(H)$.
 \item An {\em algebra of observables} is a von Neumann algebra $\calM \subseteq \calL(H)$.
\end{itemize}
\end{definition}

Every family of observables defines an algebra of observables, namely the von Neumann algebra generated by the orthogonal projections in the range of the projection-valued measures involved. In the converse direction, it is a standard result in the theory of von Neumann algebras that every algebra of observables $\calM$ is generated by the orthogonal projections belonging to $\calM$. This justifies the slight abuse of terminology in the third and fourth items in the above definition.

Let us comment on these definitions in more detail.

\subsection{Observables}
In the above definition of an observable, we allow projection-valued measures to live on an arbitrary measurable space $(\Om,\calF)$, which we think of as the ``outcome space''. The precise definition is as follows.

\begin{definition}[Projection-valued measures]\index{projection-valued measure}
Let $H$ be a Hilbert space. A {\em projection-valued measure} in $H$, defined on a measurable space $(\Om,\calF)$,
is a mapping $$P: \calF\to \ProjH$$ that assigns to every $F\in\calF$ an orthogonal projection $P_F:=P(F)$ on $H$ such that:
\begin{enumerate}[label={\rm(\roman*)}, leftmargin=*]
 \item\label{it:PVM1} $P_\Om = I$;
 \item\label{it:PVM2} for all $h\in H$, the mapping $F\mapsto \iprod{P_F h}{h}$
 defines a measure on $(\Om,\calF)$.
\end{enumerate}
\end{definition}
Projection-valued measures enjoy the following properties (see \cite[Section 9.2]{Nee}):

\begin{itemize}
 \item $P_{\emptyset} = 0$.
 \item $P_{F_1\cap F_2} = P_{F_1} P_{F_2} = P_{F_2}P_{F_1}$ for all $F_1,F_2\in \calF$.
\end{itemize}
In particular, if $F_1,F_2\in \calF$ are disjoint,
the ranges of $P_{F_1}$ and $P_{F_2}$ are orthogonal.

By the spectral theorem, real-valued observables (that is, projection-valued measures on the real line $\R$) are in one-to-one correspondence with (possibly unbounded) selfadjoint operators; this establishes the connection with the common textbook definition of (real-valued) observables as selfadjoint operators. In the same way, projection-valued measures on the unit circle $\T$ and on the complex plane $\C$ are in one-to-one correspondence with, respectively, unitary operators and (possibly unbounded) normal operators. In each of these cases, the projection-valued measure corresponding to an operator $A$ is supported on its spectrum $\sigma(A)$. Moreover, $A$ is bounded if and only if $\sigma(A)$ is bounded.

Projection-valued measures permit us to interpret
unitary operators as observables with values in the unit circle $\T$, and (possibly unbounded) normal operators as complex-valued observables. This point of view is developed systematically in \cite[Chapter 15]{Nee}. In certain settings, it is of interest to consider other outcome sets. For instance, in the context of EPR and Bell state experiments where two observers Alice and Bob perform polarisation measurements on an entangled two-photon pair, the `correct' outcome space for the joint measurements collected by Charlie is the set $\{+,-\}\times \{+,-\}$ (see Sections \ref{subsec:Holevo-EPR},  \ref{subsec:Holevo-Bell}). Alice and Bob, who observe only their respective photon, use the outcome space $\{+,-\}$.

Every family of observables defines an algebra of observables, namely the von Neumann algebra generated by the orthogonal projections in the range of the projection-valued measures involved, and conversely every algebra of observables $\calM$ is generated by the orthogonal projections belonging to $\calM$.

\subsection{States}
Let us now comment on the definitions of states and pure states. First of all, our definition of ``state'' corresponds to the notion of ``normal state'' in the literature on operator algebras.
If $\varrho\in \SH$ is a state, then
\begin{align}\label{eq:normal-state}
\om(a):= \tr(a \varrho)
\end{align}
defines a normal positive functional $\om:\calL(H)\to \C$ satisfying $\om(I) = 1$. Conversely, if $\om:\calL(H)\to \C$ is a normal positive functional satisfying $\om(I) = 1$,
there exists a unique $\varrho\in \SH$ such that \eqref{eq:normal-state} holds.

The sets of states and pure states on the Hilbert space $H$ are denoted by
$\SH$ and $\ESH$, respectively.
It is a standard result that a state is pure if and only if it is a {\em vector state}, i.e., a rank one projection of the form
$$h\,\bar \otimes\, h: \ h'\mapsto \iprod{h'}{h}h,$$
where $h\in H$ is a unit vector. Note that $h$ and its scalar multiple $ch$ define the same pure state if $|c|=1$. Thus we may identify pure states with unit vectors in $H$, provided we identify
$h$ and $ch$ when $|c|=1$. In the physics literature, the notation
$$\ket{h}$$
is used to denote the equivalence class of $h$; the rank one projection $h\,\bar\otimes\, h'$ is usually denoted as
$$ \ket{h}\bra{h}.$$
In this context it is always implicit that the vector $h$ has norm one.

\section{Indiscernibility}\label{sec:indiscernible}

The informal approach to indiscernibility of the Introduction can be made precise as follows.

\begin{definition}[Indiscernibility]\label{def:indiscernible}
Let $\varrho,\varrho'\in \SH$ be states.
\begin{itemize}
\item If $P=\{P^{(i)}:\,i\in I\}$ is a family of projection-valued measures on measurable spaces $(\Om^{(i)},\calF^{(i)})$, the states $\varrho$ and $\varrho'$ are called {\em indiscernible for $P$}, notation $$ \varrho \, \sim_{P} \, \varrho',$$
if for all $i\in I$ and $F\in \calF^{(i)}$ we have
\begin{align}\label{def-indisc-PVM}
\tr \bigl(P^{(i)}_F \varrho\bigr) =  \tr \bigl(P^{(i)}_F \varrho'\bigr).
\end{align}

\item If $\calM$ is an algebra of observables, the states $\varrho$ and $\varrho'$ are called {\em indiscernible for $\calM$}, notation
$$  \varrho \, \simM  \, \varrho',$$
if for all $a\in \calM$ we have
\begin{align}\label{def-indisc-M}
\tr (a \varrho) = \tr (a\varrho').
\end{align}
\end{itemize}
\end{definition}

If $\ket{h}$ is a pure state associated with the rank one projection $h\,\bar\otimes\,h\in \SH$, then
$ \tr (a \circ (h\,\bar\otimes\, h)) = \iprod{ah}{h}$. Accordingly,
for pure states $\ket{h}$ and $\ket{h'}$, \eqref{def-indisc-PVM} and \eqref{def-indisc-M}
reduce to
$$ \iprod{P^{(i)}_F h}{h} =  \iprod{P^{(i)}_F h'}{h'}$$
and
$$ \iprod{ah}{h} = \iprod{ah'}{h'},$$
respectively.

\subsection{First examples}
We discuss various interesting examples in Section \ref{sec:examples}. For now, we content ourselves with the two most trivial ones:

\begin{example}\label{ex:indist-min}
Every two distinct states are indiscernible for $\calM = \C I$.
\end{example}

The physics interpretation is that if we are only able to ask the question ``Is there a state?'', we cannot tell distinct states apart.

The interpretation of the next example is that if, for every pure state $\ket{h}$, we can answer the question ``Are you the pure state $\ket{h}$?'', we can tell any pair of different pure states apart.

\begin{example}\label{ex:indist-max}
If $\calM$ is a von Neumann algebra containing all rank--one orthogonal projections, then $\calM = \calL(H)$. We claim that the states $\varrho$ and $\varrho'$ are indiscernible for $\calM$ if and only if they are equal: $\varrho =\varrho'$.
The `if' part is trivial, and the `only if' part follows from the fact that the space $\calT(H)$ of all trace class operators is the dual of $\calK(H)$ by trace duality.
\end{example}

\begin{example}[Indiscernibility with respect to a spin observable]\label{ex:spin}
Let $A$ be a selfadjoint operator on a Hilbert space $H$ with spectrum
$$\sigma(A)=\{-1,1\}.$$
Then the pure states $\ket{h}$ and $\ket{h'}$ are indiscernible under the PVM associated with $A$ if and only if
$$
\iprod{Ah}{h} = \iprod{Ah'}{h'}.
$$
The `only if' part is clear.
To prove the `if' part, let $P$ be the PVM associated with $A$. The assumptions imply that $P$ is supported on $\{-1,1\}$. Setting
$p(h):=\iprod{P_{\{1\}}h}{h}$, we have
$\iprod{P_{\{-1\}}h}{h} = 1-p(h)$ and therefore
$$\iprod{A h}{h} = (+1)\,p(h) + (-1)\,[1-p(h)] = 2p(h)-1.$$
Thus if $\iprod{Ah}{h} = \iprod{Ah'}{h'}$, then $p(h) = p(h')$ and $\iprod{P_{\{\pm 1\}}h}{h} = \iprod{P_{\{\pm 1\}}h'}{h'}$.
\end{example}

Let $\calM'$ denote the {\em commutant} of $\calM$, that is,
$$\calM':= \{b\in \calL(H):\, ab = ba\, \hbox{ for all } \,a\in \calM\}.$$
We further write, for $h\in H$, $$\calM h := \{ah:\, a\in\calM\}.$$

\subsection{Indiscernibility for projection-valued measures}\label{subsec:indisc-PVM}

Let $P$ be a PVM on $(\Om,\calF)$, and let $\calM$ denote the von Neumann algebra generated by the projections in the range of $P$.
The following fact will be used repeatedly without further comment.

\begin{proposition}\label{prop:PVM-vN-indisc}
Let $P:\calF\to\Proj(H)$ be a projection-valued measure on $(\Omega,\calF)$ and let
$\calM$ be the von Neumann algebra generated by its range.
For states $\varrho,\varrho'\in\SH$ the following are equivalent:
\begin{enumerate}[\rm(i)]
\item $\varrho$ and $\varrho'$ are indiscernible for $P$;
\item $\varrho$ and $\varrho'$ are indiscernible for $\calM$.
\end{enumerate}
\end{proposition}

\begin{proof}
(ii)$\Rightarrow$(i): immediate since $P_F\in\calM$ for every $F\in\calF$.

\smallskip
(i)$\Rightarrow$(ii): Let $\calA$ be the unital $*$-algebra generated by $\{P_F:\,F\in\calF\}$.
Since the projections $P_F$ commute and $P_{F_1}\cdots P_{F_n}=P_{\cap_{j=1}^n F_j}$,
every $b\in\calA$ is a finite linear combination of projections from the range of $P$.
By linearity and assumption (i) we obtain
\begin{equation}\label{eq:expectation-on-A-mixed}
\tr(a\varrho)=\tr(a\varrho')\quad\text{for all }a\in\calA.
\end{equation}

Now fix $a\in\calM$. By Kaplansky's density theorem there exists a net $(a_\alpha)\subseteq\calA$
with $\sup_{\alpha}\|a_\alpha\|\le \|a\|$ and $\lim_{\alpha} a_\alpha = a$ strongly.
We claim that for every trace class operator $S\in\calT(H)$,
\begin{equation}\label{eq:trace-strong}
\tr(a_\alpha S)\to \tr(aS).
\end{equation}
Assuming this claim for the moment, we apply it with $S=\varrho$ and $S=\varrho'$.
From \eqref{eq:expectation-on-A-mixed} we have $\tr(a_\alpha\varrho)=\tr(a_\alpha\varrho')$ for all $\alpha$,
and passing to the limit using \eqref{eq:trace-strong} gives $\tr(a\varrho)=\tr(a\varrho')$.
Since $a\in\calM$ was arbitrary, (ii) follows.

\smallskip
Proof of the claim \eqref{eq:trace-strong}: Fix $S\in\calT(H)$ and $\eps>0$.
Choose a finite rank operator $S_0$ such that $\|S-S_0\|_1<\eps/(4\|a\|)$, where $\|\cdot\|_1$ is the trace norm; such an operator exists by the density of finite rank operators.
Then, using that $|\tr(TU)|\le \|T\|\,\|U\|_1$, we have
\[
|\tr((a_\alpha-a)(S-S_0))|
\le \|a_\alpha-a\|\,\|S-S_0\|_1
\le 2\|a\|\,\|S-S_0\|_1
< \eps/2.
\]
Next, write $S_0$ as a finite sum of rank--one operators,
$S_0=\sum_{k=1}^n u_k\,\bar\otimes\, v_k$. Then $\tr(T(u_k\,\bar\otimes\, v_k))=(Tu_k|v_k)$
and
\[
\tr((a_\alpha-a)S_0)=\sum_{k=1}^n ( (a_\alpha-a)u_k \,|\, v_k).
\]
Since $a_\alpha\to a$ strongly, there exists $\alpha_0$ such that
$|\tr((a_\alpha-a)S_0)|<\eps/2$ for all $\alpha\ge \alpha_0$.
Combining the two estimates gives $|\tr((a_\alpha-a)S)|<\eps$ for $\alpha\ge\alpha_0$, thus proving \eqref{eq:trace-strong}.
\end{proof}

When the von Neumann algebra $\calM$ is generated by the projections in a family of PVMs $P=\{P^{(i)}:\, i\in I\}$ on measurable spaces $(\Om^{(i)},\calF^{(i)})$, it is tempting to test indiscernibility for $\calM$ by just checking the equalities
$$\iprod{P^{(i)}_F h}{h} = \iprod{P^{(i)}_F h'}{h'}, \quad F\in \calF^{(i)}$$ for all individual PVMs $P^{(i)}$ in this family.
However, this only controls the marginal measurement statistics of each PVM separately. Indiscernibility for $\calM$ also requires
the equality
$$
 \iprod{P^{(i_1)}_{F_1}\cdots P^{(i_n)}_{F_n}h}{h}
 = \iprod{P^{(i_1)}_{F_1}\cdots P^{(i_n)}_{F_n}h'}{h'},$$
for all $i_1,\dots,i_n\in I$ with $F_k\in \calF^{(i_k)}$ for all $k=1,\dots,n$.
Physically, if all PVMs commute the above products correspond to joint events; for noncommuting families, these products are merely elements of the generated algebra and should not be interpreted as joint outcome probabilities.

As the following example shows,
two pure states may be indiscernible for a family of (commuting) PVMs, but discernible for the (abelian) von Neumann algebra $\calM$ generated by the projections in this family: the states may still be discerned by $\calM$ through a product of projections.

\begin{example}
\label{ex:counterexample_prop_indisc_specmeas}
Let $H=\C^2\otimes\C^2$, and let $P_3$ be the PVM of the {\em Pauli-$z$ matrix}
$$\sigma_3 = \begin{pmatrix}\, 1 & \, 0\, \\ 0 & -1 \end{pmatrix}. $$
Consider the commuting PVMs on $\{-1,1\}$, the spectrum of $\sigma_3$, which is uniquely defined by specifying their values on  singletons by
\[
P_A(\{\eps\}) := P_3(\{\eps\})\otimes I, \qquad
P_B(\{\eps\}) := I\otimes P_3(\{\eps\}), \qquad \eps\in\{-1,1\},
\]
where $P_3(\{1\}) = \ket{0}\bra{0}$ and $P_3(\{-1\}) = \ket{1}\bra{1}$,
and let $\calM$ be the von Neumann algebra generated by their ranges (the diagonal
algebra with respect to the standard basis $\{\ket{00},\ket{01},\ket{10},\ket{11}\}$ of $\C^2\otimes\C^2$).

Define the two Bell states
\[
\ket{h}:=\frac{1}{\sqrt2}(\ket{00}+\ket{11}), \qquad
\ket{h'}:=\frac{1}{\sqrt2}(\ket{01}+\ket{10}).
\]
Then for each $\eps\in\{-1,1\}$ we have
\begin{align*}
\iprod{P_A(\{\eps\})h}{h}
& =
\iprod{P_A(\{\eps\})h'}{h'}
=\frac12,
\\
\iprod{P_B(\{\eps\})h}{h}
& =
\iprod{P_B(\{\eps\})h'}{h'}
=\frac12,
\end{align*}
so $h$ and $h'$ have identical statistics for each observable separately.

However, the product of commuting projections
$P_A(\{1\})P_B(\{1\})=P_3(\{1\})\otimes P_3(\{1\})=\ket{00}\bra{00}$
belongs to $\calM$, and
\[
\iprod{(\ket{00}\bra{00})h}{h}=\frac12,
\qquad
\iprod{(\ket{00}\bra{00})h'}{h'}=0.
\]
Hence $\ket{h}$ and $\ket{h'}$ are discernible for $\calM$.
\end{example}

This example will be pivotal in the final section on the EPR and Bell experiments.
In the EPR experiment, the polarisation measurements by Alice and Bob on their respective photons can be modelled as the PVMs $P_A$ and $P_B$.
As long as Alice and Bob have access only to their separate marginal statistics, they cannot distinguish pure states $\ket{h}$ and $\ket{h'}$ on the basis of their measurement statistics if and only if $\ket{h}$ and $\ket{h'}$ are indiscernible for the pair $\{P_A,P_B\}$, i.e.,
$$\iprod{P_A(\{\eps\})h}{h} = \iprod{P_A(\{\eps\})h'}{h'}, \quad
\iprod{P_B(\{\eps\})h}{h} = \iprod{P_B(\{\eps\})h'}{h'},
$$ for all $\eps \in \{-1,1\}$.
The von Neumann algebra generated by ranges of $P_A$ and $P_B$ contains the range of the product PVM $P_C = P_A\times P_B$ which an external observer Charlie can use to analyse the measurement statistics of the {\em pairs of outcomes} generated by Alice and Bob. The above example shows that Charlie can distinguish certain pairs of pure states that are indiscernible for Alice and Bob on the basis of their measurement statistics.

\subsection{Indiscernibility for algebras of observables}

The following result states that the equivalence classes of pure states under indiscernibility for an algebra of observables $\calM$ are precisely the 'orbits' under the set of partial isometries commuting with the elements of $\calM$. An elementary proof of this fact (which can alternatively be deduced from the uniqueness of GNS representations) is included for the convenience of the reader.

Recall that an operator $V \in \calL(H)$ is a {\em partial isometry} if there is an orthogonal decomposition $H = H_0\oplus H_1$ such that the restricted operator $V|_{H_0}$ is an isometry (i.e., $\n Vh_0\n = \n h_0\n$ for all $h_0\in H_0$)
and $V|_{H_1} = 0$. In this situation, $H_0$ is called the {\em initial space} of $V$ and the (closed) range of $V|_{H_0}$ is called the {\em final space} of $V$.

\begin{theorem}[Indiscernibility and partial isometries]\label{thm:indisc-partialiso}
Let $\calM$ be a von Neumann algebra acting on $H$. For unit vectors $h,h'\in H$
the following assertions are equivalent:
\begin{enumerate}[\rm(1)]
 \item\label{it:indisc-psi} $\iprod{ah}{h}=\iprod{ah'}{h'}$ for all $a\in\calM$;
 \item\label{it:partialiso} there exists a partial isometry $V\in\calM'$ such that
 $Vh=h'$.
\end{enumerate}
In this situation, the operator $V$ can be chosen to satisfy
$$ V^*V=P_{\overline{\calM h}},\qquad VV^*=P_{\overline{\calM h'}},$$
the orthogonal projections in $H$ onto $\overline{\calM h}$ and $\overline{\calM h'}$, respectively.
\end{theorem}

\begin{proof}
\ref{it:indisc-psi}$\Rightarrow$\ref{it:partialiso}:\
Assume that $\iprod{ah}{h}=\iprod{ah'}{h'}$ for all $a\in\calM$.
For all $a\in\calM$ we have $a^*a\in\calM$ and consequently
\[
\|ah\|^2=\iprod{a^*ah}{h}=\iprod{a^*ah'}{h'}=\|ah'\|^2,
\]
so that the linear mapping $\calM h\to \calM h'$ given by $ah\mapsto ah'$ is
well-defined and isometric, and by continuity it extends uniquely to a unitary
operator $U_0$ from $\overline{\calM h}$ onto $\overline{\calM h'}$ satisfying
$$ U_0(ah)=ah'\ \ \text{for all }a\in\calM.$$
In particular, $U_0h=h'$.
Moreover, for every $a\in\calM$ and every $x\in \overline{\calM h}$ we have
the intertwining relation
\begin{equation}\label{eq:U0-intertwine}
U_0ax=aU_0x.
\end{equation}
Indeed, this is clear if $x = bh \in \calM h$, and this extends to general $x\in \overline{\calM h}$ by approximation.

Next, we check that $\overline{\calM h}$ and $\overline{\calM h'}$ are invariant under $\calM$.
The invariance of $\overline{\calM h}$ is clear since $\calM h$ is $\calM$-invariant.
If $g\in \overline{\calM h}^\perp$ and $b\in\calM$, then for all $a\in\calM$,
$$\iprod{ah}{bg}=\iprod{b^*ah}{g}=0$$
since $b^*a\in\calM$ and $g\perp \calM h$,
Hence, $bg\in \overline{\calM h}^\perp$. This shows that also $\overline{\calM h}^\perp$ is invariant under $\calM$. It follows that $P_{\overline{\calM h}}\in\calM'$. The same argument applies to $\overline{\calM h'}$, hence also $P_{\overline{\calM h'}}\in\calM'$.

Now, define an operator $V\in\calL(H)$ by
\[
V:=U_0 P_{\overline{\calM h}}.
\]
Then $V$ is a partial isometry with initial space $\overline{\calM h}$ and final space $\overline{\calM h'}$, so
\[
V^*V=P_{\overline{\calM h}},\qquad VV^*=P_{\overline{\calM h'}}.
\]
Moreover, $Vh=U_0h=h'$. The identity $V^*V=P_{\overline{\calM h}}$ then implies that also $V^* h'=h$.

It remains to show that $V\in\calM'$. Fix $a\in\calM$ and $x\in H$. Using that
$P_{\overline{\calM h}}\in\calM'$, we have $P_{\overline{\calM h}}ax=aP_{\overline{\calM h}}x\in \overline{\calM h}$, so by \eqref{eq:U0-intertwine},
\[
Va
=U_0 P_{\overline{\calM h}}a
=U_0 a P_{\overline{\calM h}}
=a U_0 P_{\overline{\calM h}}
=aV.
\]
Thus $Va=aV$ for all $a\in\calM$, i.e., $V\in\calM'$.

\smallskip
\ref{it:partialiso}$\Rightarrow$\ref{it:indisc-psi}:\
Let $V\in\calM'$ be a partial isometry with $Vh=h'$.
Since $V$ is a partial isometry,
there exists an orthogonal decomposition $H=H_0\oplus H_1$ such that
$V|_{H_0}$ is an isometry and $V|_{H_1}=0$. Let $P$ denote the orthogonal projection
of $H$ onto $H_0$. For $x,y\in H$ write $x=x_0+x_1$, $y=y_0+y_1$ with $x_i,y_i\in H_i$.
Since $Vx=Vx_0$ and $Vy=Vy_0$, and $V|_{H_0}$ is an isometry, we obtain
\[
\iprod{Vx}{Vy}=\iprod{Vx_0}{Vy_0}=\iprod{x_0}{y_0}=\iprod{Px}{y}.
\]
and therefore $V^*V= P$. In particular, $I-V^*V\ge 0$. Moreover,
\[
\iprod{(I-V^*V)h}{h} =\iprod{h}{h}-\iprod{V^*Vh}{h} = \|h\|^2-\|Vh\|^2.
\]
But $\|Vh\|=\|h'\|=1=\|h\|$, so $\iprod{(I-V^*V)h}{h}=0$. By positivity, this implies $(I-V^*V)h=0$, i.e.\ $V^*Vh=h$.
Now for every $a\in\calM$, using $V^*V\in\calM'$ we obtain
\[
V^*V(ah)=a(V^*Vh)=ah,
\]
so $V^*V$ acts as the identity on $\calM h$, and therefore on $\overline{\calM h}$.
In particular, $V^*h'=V^*Vh=h$.

For any $a\in\calM$, using $V\in\calM'$ we then obtain
\[
\iprod{ah'}{h'}
=\iprod{ah'}{Vh}
=\iprod{V^*ah'}{h}
=\iprod{aV^*h'}{h}
=\iprod{ah}{h}.
\]
This proves indiscernibility.
\end{proof}

If $h$ and $h'$ are {\em cyclic} for $\calM$ (which means that $\calM h$ and $\calM h'$ are dense in $H$), then $P_{\overline{\calM h}}=P_{\overline{\calM h'}}=I$, so the partial isometry $V$
satisfies $V^*V=VV^*=I$ and $V$ is unitary.
The following uniqueness result implies that in this situation the unitary operator $V$ is also unique:

\begin{proposition}[Uniqueness on $\overline{\calM h}$]\label{prop:partial-uniqueness}
Let $\calM$ be a von Neumann algebra acting on $H$, and let $h,h'\in H$ be unit
vectors.
Then every partial isometry $U\in\calM'$ satisfying $Uh=h'$ maps $\overline{\calM h}$ onto
$\overline{\calM h'}$.
Moreover, if $U,V\in\calM'$ are partial isometries satisfying
\[Uh=h'=Vh,\]
then
\[U|_{\overline{\calM h}}=V|_{\overline{\calM h}}.\]
\end{proposition}

\begin{proof}
Let $U\in\calM'$ be a partial isometry with $Uh=h'$. For every $a\in\calM$ we have $U(ah) = aUh = ah'$, so $U(\calM h)=\calM h'$.

Since $U$ is a partial isometry, $P:=U^*U$ is an orthogonal projection (the initial
projection of $U$), and hence $I-P$ is an orthogonal projection. Moreover,
\[
\|(I-P)h\|^2 = \iprod{(I-P)h}{h}=\|h\|^2-\|Uh\|^2=1-1=0.
\]
This implies $(I-P)h=0$, i.e.\ $Ph=h$. Since $P\in\calM'$ and $a\in\calM$,
we obtain
\[
P(ah)=a(Ph)=ah \qquad (a\in\calM),
\]
so $P$ acts as the identity on $\calM h$, and therefore on $\overline{\calM h}$.
In particular, $\overline{\calM h}\subseteq \Ran(P)$, and $U$ restricts to an isometry on
$\overline{\calM h}$. Hence $U(\overline{\calM h})$ is closed. Together with
$U(\calM h)=\calM h'$ this yields
\[
U(\overline{\calM h})
=\overline{U(\overline{\calM h})}
=\overline{U(\calM h)}
=\overline{\calM h'}.
\]
This proves that $U$ maps $\overline{\calM h}$ onto
$\overline{\calM h'}$.

Now, let $U,V\in\calM'$ be partial isometries with $Uh=Vh=h'$. Then for every $a\in\calM$ we have $U(ah) = aUh = ah' = aVh = V(ah).$ Therefore, $U$ and $V$ agree on the dense subspace $\calM h$ of $\overline{\calM h}$,
and by continuity, they agree on $\overline{\calM h}$.
\end{proof}

The construction in the proof of Theorem \ref{thm:indisc-partialiso}
always gives a partial isometry $V\in\calM'$ with $V^*V=P_{\overline{\calM h}}$ and $VV^*=P_{\overline{\calM h'}}$.
Such $V$ can be extended to a \emph{unitary} $U\in\calM'$ with $Uh=h'$ provided
there exists a partial isometry $W\in\calM'$ with
\[
W^*W=I-P_{\overline{\calM h}},\qquad WW^*=I-P_{\overline{\calM h'}}.
\]
Indeed, the initial projections $V^*V$ and $W^*W$ are orthogonal and sum to $I$,
and similarly the final projections $VV^*$ and $WW^*$ are orthogonal and sum to $I$.
Hence,
\[
V^*W=(V^*V)V^*W(W^*W)=0,\qquad
VW^*=(VV^*)VW^*(WW^*)=0,
\]
and similarly $W^*V=0$ and $WV^*=0$.
Therefore, with $U:=V+W$ we have
\[
U^*U = V^*V + W^*W = I,\qquad
UU^* = VV^* + WW^* = I,
\]
so $U$ is unitary. Since $\calM'$ is a $*$-algebra, $U\in\calM'$.
Finally, $Uh=Vh=h'$ because $h\in\overline{\calM h}$ and $W$ vanishes on
$\overline{\calM h}$.

For indiscernible $h,h'\in H$, a unitary $U\in \calM'$ satisfying $Uh=h'$ need not always exist.
This will be shown in Example \ref{ex:no-unitary}, the construction of which relies on the following lemma.

\begin{lemma}\label{lem:support-conjugation}
Let $\calM\subseteq\calL(H)$ be a von Neumann algebra. For all $h\in H$ we have
$P_{\overline{\calM h}} \in \calM'$.
If $u\in\calM'$ is unitary and $h':=uh$, then
\[
u\,P_{\overline{\calM h}}\,u^* = P_{\overline{\calM h'}}.
\]
\end{lemma}

\begin{proof}
As observed in the proof of Theorem \ref{thm:indisc-partialiso}, the subspace $\overline{\calM h}$ is $\calM$-invariant, and so is its orthogonal complement. Therefore $P_{\overline{\calM h}}\in\calM'$.

If $u\in\calM'$ is unitary, then $u(\calM h)=\calM(uh)$ because
$u(ah)=a(uh)$ for all $a\in\calM$. Taking closures, it follows that
$u(\overline{\calM h})=\overline{\calM(uh)}=\overline{\calM h'}$.
Thus $uP_{\overline{\calM h}}u^*$ is the orthogonal projection onto
$u(\overline{\calM h})=\overline{\calM h'}$.
\end{proof}

In combination with Theorem \ref{thm:indisc-partialiso}, this lemma shows that if $P_{\overline{\calM h}}$ and $P_{\overline{\calM h'}}$ are not unitarily
conjugate by a unitary $u\in \calM'$, then there is no unitary $u\in\calM'$ with $uh=h'$.
In particular, since dimensions of kernels are preserved by unitaries, this situation arises when the kernels of $P_{\overline{\calM h}}$ and $P_{\overline{\calM h'}}$ are of different dimensions. We will use this to show that Theorem \ref{thm:indisc-partialiso} is wrong if partial isometries are replaced by unitaries:

\begin{example}\label{ex:no-unitary}
Let $H_1=H_2=\ell^2(\N )$ and let $\calM:=\calL(H_1)\otimes I_{H_2}$
acting on $H:=H_1\otimes H_2$. Then $\calM'=I_{H_1}\otimes \calL(H_2)$.

Let $(e_n)_{n\in\N }$ be the standard basis of $\ell^2(\N )$ and define
unit vectors $h,h'\in H_1\otimes H_2$ by
\[
h:=c\sum_{n\in\N }\frac1{n+1}\, e_n\otimes e_{n+1},
\qquad
h':=c\sum_{n\in\N }\frac1{n+1}\, e_n\otimes e_{n+2},
\]
where $c=1/(\sum_{n\in\N }\frac1{(n+1)^2})^{1/2}$.
For every $A\in\calL(H_1)$ we have, by orthonormality,
\[
\iprod{(A\otimes I)h}{h}
= c^2\sum_{n\in\N }\frac1{(n+1)^2}\iprod{Ae_n}{e_n}
=\iprod{(A\otimes I)h'}{h'},
\]
so $h$ and $h'$ are indiscernible for $\calM$.

Let $S$ be the unilateral shift on $H_2$, i.e., $Se_n=e_{n+1}$ for $n\in\N$.
Set
\[
P_1:=SS^* = I-e_0\,\bar\otimes\, e_0,
\qquad
P_2:=S^2(S^2)^* = I-e_0\,\bar\otimes\, e_0-e_1\,\bar\otimes\, e_1.
\]
Then $v:=SP_1$ is a partial isometry on $H_2$ with $v^*v=P_1$ and $vv^*=P_2$.
Consequently $V:=I\otimes v\in\calM'$ is a partial isometry with
$V^*V=I\otimes P_1$ and $VV^*=I\otimes P_2$, and one checks directly that
\[
(I\otimes v)h=h'.
\]

However, there is no unitary $U\in\calM'$ with $Uh=h'$.
Indeed, any unitary in $\calM'$ has the form $U=I\otimes u$ with $u$ unitary on $H_2$.
If $Uh=h'$, then by Lemma~\ref{lem:support-conjugation},
\[
(I\otimes u)P_{\overline{\calM h}}(I\otimes u)^* = P_{\overline{\calM h'}}.
\]
We claim that
\[
P_{\overline{\calM h}}=I\otimes P_1,
\qquad
P_{\overline{\calM h'}}=I\otimes P_2.
\]
To see the first identity, note that for each $n\in\N $ and $x\in H_1$ the rank--one
operator $A_{x,n}:=x\,\bar\otimes\, e_n : y\mapsto \iprod{y}{e_n}x$ on $H_1$ satisfies
\[
(A_{x,n}\otimes I)h = \frac{c}{n+1}\, x\otimes e_{n+1}.
\]
Hence $H_1\otimes e_{n+1}\subseteq \calM h$ for all $n$, so
$H_1\otimes \Ran(P_1)\subseteq \overline{\calM h}$.

Conversely, note that $h\in H_1\otimes \Ran(P_1)=\Ran(I\otimes P_1)$, hence
\[
(I\otimes(I-P_1))h=0.
\]
For any $A\in\calL(H_1)$ we have $A\otimes I\in\calM$ and
\((A\otimes I)(I\otimes(I-P_1))=(I\otimes(I-P_1))(A\otimes I)\).
Therefore,
\[
(I\otimes(I-P_1))(A\otimes I)h
=(A\otimes I)(I\otimes(I-P_1))h
=0.
\]
Thus $(A\otimes I)h\in\Ker(I\otimes(I-P_1))=\Ran(I\otimes P_1)=H_1\otimes\Ran(P_1)$.
Hence $\calM h\subseteq H_1\otimes\Ran(P_1)$ and therefore also
\(\overline{\calM h}\subseteq H_1\otimes\Ran(P_1).\)
This proves the claim for $P_1$.
The argument for $h'$ is identical, with $P_1$ replaced by $P_2$.

Using the identifications
\(
P_{\overline{\calM h}}=I\otimes P_1
\)
and
\(
P_{\overline{\calM h'}}=I\otimes P_2
\)
just proved, the conjugation identity becomes
\[
(I\otimes u)(I\otimes P_1)(I\otimes u)^* \, =\,  I\otimes P_2.
\]
But
\[
(I\otimes u)(I\otimes P_1)=I\otimes (uP_1)
\quad\text{and}\quad
(I\otimes u)^*=I\otimes u^*,
\]
hence
\[
(I\otimes u)(I\otimes P_1)(I\otimes u)^*
= (I\otimes uP_1)(I\otimes u^*)
= I\otimes (uP_1u^*).
\]
Therefore \(I\otimes (uP_1u^*)=I\otimes P_2\), and thus \(uP_1u^*=P_2\).
But $\dim\ker P_1=1$ and $\dim\ker P_2=2$, and unitary conjugation preserves the
dimension of kernels. This contradiction shows that the role of partial isometries cannot be replaced by unitaries in Theorem \ref{thm:indisc-partialiso}.
\end{example}

We conclude with a result that will be needed in Section \ref{sec:EPR_and_Bell}.

\begin{theorem}[Unitary invariance of the Holevo space and equivalence classes]
\label{thm:Holevo-invariance-under-rotation}
Let $P:\calF\to\Proj(H)$ be a projection-valued measure, and let $U$ be a unitary on $H$.
Define the conjugated PVM $P^U:\calF\to\Proj(H)$ by
\[
P^U(F):=U^*P(F)U,\qquad F\in\calF.
\]
The mapping
\[
\Phi_U\big([h]_{P^U}\big):=[Uh]_{P},
\]
is well defined and bijective, and we have
$$\ket{h}\sim_{\! P^U} \ket{h'} \ \Longleftrightarrow \ Uh\sim_{\!P} Uh'.$$
Furthermore, we have
\[\calM^U=U^*\calM U,\]
where $\calM$ and $\calM^U$ are the von Neumann algebras generated
by the ranges of $P$ and $P^U$, respectively.
\end{theorem}

\begin{proof}
To prove well-definedness, suppose that $[h]_{P^U}=[h']_{P^U}$, i.e.,
\[
\iprod{P^U(F)h}{h}=\iprod{P^U(F)h'}{h'}, \qquad F\in \calF.
\]
Fix $F\in\calF$. Then
\begin{align*}
\iprod{P(F)\,Uh}{Uh}
& =\iprod{U^*P(F)U\,h}{h}
=\iprod{P^U(F)h}{h}
\\ & =\iprod{P^U(F)h'}{h'}
=\iprod{U^*P(F)U\,h'}{h'}
=\iprod{P(F)\,Uh'}{Uh'}.
\end{align*}
Since this holds for all $F\in\calF$, we conclude $Uh \sim_{\!P}Uh'$, so
$[Uh]_{P}=[Uh']_{P}$ and $\Phi_U$ is well defined.
Injectivity and surjectivity follow from the fact that $U$ is invertible.

The identity $\calM^U = U^*\calM U$ is immediate from the definition, as the range of $P^U$ is $\{U^*P(F)U:\ F\in \calF\}$ and conjugation with unitaries preserves strong closures. Hence, the von Neumann algebra it generates is the conjugate by $U$ of the von Neumann algebra generated by $P$.
\end{proof}

\section{The Holevo space}\label{sec:Holevo}

The following definition goes back to Holevo \cite{Hol}. We recall the equivalence relations $\simP$ and $\simM$ of Definition \ref{def:indiscernible}.

\begin{definition}[Holevo space] \

\begin{itemize}
\item If $P = \{P^{(i)}:\, i\in I\}$ is a family of observables,
the quotient space of $\SH$ modulo $\simP$
is called the {\em Holevo space} of $H$ associated with $P$, and is denoted by $$\SHP.$$
\item If $\calM$ is an algebra of observables,
the quotient space of $\SH$ modulo $\simM$
is called the {\em Holevo space} of $H$ associated with $\calM$, and is denoted by $$\SHM.$$
\end{itemize}
Restricting the equivalence relations to the set of pure states $\ESH$, we obtain the {\em pure Holevo spaces} of $H$ associated with $P$ and $\calM$. These are denoted respectively by $\ESHP$ and $\ESHM.$
\end{definition}

If $\calM$ is the von Neumann algebra generated by the range of a PVM $P$, Proposition \ref{prop:PVM-vN-indisc} implies that
\begin{align}\label{eq:PVM-vN-indisc} \SHP = \SHM \ \ \ \hbox{and} \ \ \ \ESHP = \ESHM.
\end{align}

The equivalence classes of a state $\varrho\in \SH$ with respect to these equivalence relations will be denoted by $[\varrho]_P$ and $[\varrho]_\calM$ respectively.
With slight abuse of notation, the equivalence classes of a pure state $\ket{h}$ will be denoted by $[h]_P$ and $[h]_\calM$.

The physics intuition is to think of the elements of $P$ and $\calM$ as the available ``observables''. The elements of the spaces $\SHP$ and $\SHM$ are then the ``states relative to $P$ and $\calM$'', respectively. Informally, they are the ``distinguishable objects'': an observer who has access only to the observables of $P$ or $\calM$ is not able to distinguish between states that are indiscernible for $P$ or $\calM$.

The following proposition identifies $\SHM$ with the set $\SM$ of normal states on $\calM$,
i.e., the set of all normal positive functionals $\om:\calM\to \C$ satisfying $\om(I)=1$. Every such $\om\in\SM$ admits a representation
\[
\om(a)=\tr(a\varrho),\qquad a\in\calM,
\]
for some positive trace class operator $\varrho$ on $H$ with $\tr(\varrho)=1$; see \cite[Theorem 2.4.21]{BraRob} or \cite[Corollary 4.3]{Landsman}.
In general, the representing operator $\varrho$ need not be unique.

\begin{proposition}[Identification of $\SHM$ with normal states on $\calM$]\label{prop:SM-identification}
For $\varrho\in\SH$ define $\omega_\varrho:\calM\to\C$ by
$$\omega_\varrho(a):=\tr(a\varrho),\qquad a\in\calM.$$
Then $\omega_\varrho\in\SM$. Moreover, the mapping $\varrho\mapsto \omega_\varrho$ induces an affine (i.e., convexity preserving) bijection
\[\Phi:\SHM\ \cong \  \SM.\]
\end{proposition}

\begin{proof}
If $\varrho\simM \varrho'$, then $\tr(a\varrho)=\tr(a\varrho')$ for all $a\in\calM$,
so $\omega_\varrho=\omega_{\varrho'}$. Hence the induced mapping $[\varrho]_\calM\mapsto \omega_\varrho$ is well defined. Clearly, this mapping preserves convex combinations.
If $\om_\varrho = \om_{\varrho'}$, then $\tr(a\varrho)=\tr(a\varrho')$ for all $a\in\calM$,
so $\varrho\simM \varrho'$,
and therefore $[\varrho]_{\calM}=[\varrho']_{\calM}$. It follows that the mapping $[\varrho]_\calM\mapsto \omega_\varrho$ is injective.
To prove surjectivity, let $\omega\in\SM$ and choose $\varrho\in \SH$ such that
\(\tr(a\varrho)=\omega(a),\) for all $a\in \calM$. Then $[\varrho]_{\calM}$ maps to $\omega$.
\end{proof}

\subsection{The topology $\tauM$}

In this section, we consider a fixed algebra of observables $\calM$ and study the structure of the Holevo spaces $\SHM$ and $\ESHM$ in more detail. In view of \eqref{eq:PVM-vN-indisc}, this also covers the Holevo spaces associated with a single PVM.

For every $a\in\calM$ we consider the mapping $f_a: \SH \to \C$ defined by $$ f_a(\varrho) := \tr(a\varrho).$$

\begin{definition}[The topologies $\tauM$ and $\tauMt$] The topology $\tauM$ on
$\SH$ is defined as the coarsest topology making $f_a$ continuous for every $a\in\calM$. The induced quotient topology on $\SHM$, i.e., the finest topology for which the quotient mapping is continuous, is denoted by $\tauMt$.
\end{definition}

The quotient mapping from $\SH$ to $\SHM$ is denoted by $q$. It is immediate from the definition of $\tauMt$ that $q$ is continuous and that these equivalence classes are $\tauM$-closed.

\begin{lemma}\label{lem:opennes_of_quotient_TM}
    The quotient topology of $\SHM$ is precisely the collection of sets $U\subseteq \SHM $ with the property that $ q^{-1}(U)$ is $\tauM$-open in $\SH$.
\end{lemma}
\begin{proof}
We claim that the quotient topology of $\SHM$ is precisely the collection of sets $U\subseteq \SHM $ with the property that $ q^{-1}(U)$ is $\tauM$-open in $\SH$.
Indeed, denoting this collection of sets by $\upsilon$ for the moment, we note that $\upsilon$ is a topology that makes $q$ continuous.
Since $\tauMt$ is the finest topology that makes $q$ continuous, we conclude that $\upsilon \subseteq \tauMt$.
Conversely, the map $q:(\SH,\tauM)\to(\SHM,\tauMt)$ is continuous by definition of $\tauMt$.
Therefore, for every $U\in\tauMt$ the preimage $q^{-1}(U)$ is $\tauM$-open in $\SH$. But this means that $U\in\upsilon$.
Hence $\tauMt\subseteq \upsilon$.
   \end{proof}

The same argument shows that the quotient topology of $\ESHM$ is precisely the collection of sets $U\subseteq \ESHM $ with the property that $ r^{-1}(U)$ is open in $\ESH$ with respect to the subspace topology inherited from $\SH$, where $r = q|_{\ESH}: \ESH \to \ESHM$ is the quotient mapping.

\begin{lemma}\label{lem:saturation}
For every $\tauM$-open set $O\subseteq \SH$ we have
\[
q^{-1}(q(O))=O.
\]
\end{lemma}

\begin{proof}
For $a\in\calM$ and an open set $V\subseteq\C$, the lemma holds for the set
$B:=f_a^{-1}(V)$, for if $\varrho\in B$ and
$\varrho\simM\varrho'$, then $f_a(\varrho')=f_a(\varrho)\in V$, hence $\varrho'\in B$.
Thus $q^{-1}(q(B))=B$.

The collection of all open sets $O$ satisfying $q^{-1}(q(O))=O$ is closed under taking finite intersections and arbitrary unions.
Since $\tauM$ is generated by the family $\{f_a^{-1}(V)\}$ under these
operations, every $\tauM$-open set $O$ satisfies $q^{-1}(q(O))=O$.
\end{proof}

As an immediate consequence of this lemma, the quotient mapping
$$q: \SH \to \SHM $$
is open. The lemma also implies that the quotient topology is Hausdorff. The proof is standard; it is included for the reader's convenience.

\begin{lemma}
The space $(\SHM, \tauMt)$ is a Hausdorff space.
\end{lemma}

\begin{proof}
Suppose that $ [\varrho_1]_{\calM}$ and $[\varrho_2]_{\calM} $ are distinct equivalence classes in $\SHM$. By the definition of $\simM $, there exists an $a\in \calM$ for which
$f_a(\varrho_1) \neq f_a(\varrho_2).$

Choose disjoint open sets $U,V\subseteq\C$ such that
$f_a(\varrho_1)\in U$ and $f_a(\varrho_2)\in V$, and define
\begin{align*}
W_U & := \bigl\{[\varrho]_{\calM}\in\SHM:\ \exists \varrho'\in f_a^{-1}(U)\ \text{with}\ \varrho\simM\varrho'\bigr\}, \\
W_V & := \bigl\{[\varrho]_{\calM}\in\SHM:\ \exists \varrho'\in f_a^{-1}(V)\ \text{with}\ \varrho\simM\varrho'\bigr\} .
\end{align*}
We claim that $W_U$ and $W_V$ are open in $\SHM$. To prove this, we first observe that
\begin{align}\label{eq:subclaim}
q^{-1}(W_U)=f_a^{-1}(U).
\end{align}
To see this, note that if $\varrho\in f_a^{-1}(U)$ then $[\varrho]_{\calM}\in W_U$, so
$\varrho\in q^{-1}(W_U)$. Conversely, if $\varrho\in q^{-1}(W_U)$ then $[\varrho]_{\calM}\in W_U$, so there exists
$\varrho'\in f_a^{-1}(U)$ with $\varrho\simM\varrho'$. Therefore $f_a(\varrho)=f_a(\varrho')\in U$, i.e., $\varrho\in f_a^{-1}(U)$.
This proves \eqref{eq:subclaim}.

Since $f_a^{-1}(U)$ is $\tauM$-open in $\SH$,
lemma \ref{lem:opennes_of_quotient_TM} implies that
$W_U$ is open in $\SHM$. Likewise, $W_V$ is open.

Clearly, $[\varrho_1]_{\calM}\in W_U$ and $[\varrho_2]_{\calM}\in W_V$. To complete the proof, we show that $W_U\cap W_V=\emptyset$.

Arguing by contradiction, suppose that
$[\varrho]_{\calM} \in W_U \cap W_V$.
Then there exist
$\rho_1 \in f_a^{-1}(U)$ and $\rho_2 \in f_a^{-1}(V)$
such that $\varrho \simM \rho_1$ and $\varrho\simM \rho_2$, i.e.,
$$[\rho_1]_{\calM} = [\varrho]_{\calM} = [\rho_2]_{\calM}.$$
This means that $\rho_1$ and $\rho_2$ belong to the same equivalence class, and therefore
$f_a(\rho_1) = f_a(\rho_2)$.
This common value belongs to
$U$ (because $\rho_1 \in f_a^{-1}(U)$) and to $V$ (because $\rho_2 \in f_a^{-1}(V)$).
Since $U$ and $V$ are disjoint, we have arrived at a contradiction.
\end{proof}

The inclusion $\ESH \subseteq  \SH$ induces a natural inclusion mapping $$\ESHM \subseteq  \SHM$$
which sends the $\simM $-class of a pure state to its equivalence class in $\SHM$.
Explicitly, define $j:\ESHM\to \SHM$ by
\[
j\big([h]_{\calM}\big):=q(h\,\bar\otimes\, h).
\]
This mapping is well defined, for if $[h]_{\calM}=[h']_{\calM}$ in $\ESHM$, then
$\iprod{ah}{h}=\iprod{ah'}{h'}$ for all $a\in\calM$, i.e.,
$\tr(a(h\,\bar\otimes\, h))=\tr(a(h'\,\bar\otimes\,h'))$ for all $a\in\calM$, and therefore
$q(h\,\bar\otimes\, h)=q(h'\,\bar\otimes\, h')$.

The mapping $j$ is also injective,
for if $q(h\,\bar\otimes\, h) = q(h'\,\bar\otimes\, h')$, then $h\,\bar\otimes\, h \simM h'\,\bar\otimes\, h'$ as elements of $\SH$, hence
$\tr(a(h\,\bar\otimes\, h))=\tr(a(h'\,\bar\otimes\, h'))$ for all $a\in \calM$,
or equivalently
$ \iprod{ah}{h} = \iprod{ah'}{h'}$ for all $a\in \calM$, and therefore
$[h]_\calM = [h']_\calM$.

Let $r:= q|_{\ESH}:\ESH\to \ESHM$
be the quotient map, and let $\iota:\ESH \subseteq \SH$ denote the inclusion map. With this notation, we have the commuting diagram
\begin{equation}\label{eq:commute-quotients}
\begin{tikzcd}
\ESH \arrow[r, "r"] \arrow[d, "\iota"'] & \ESHM \arrow[d, "j"] \\
\SH \arrow[r, "q"'] & \SHM
\end{tikzcd}
\end{equation}
The next lemma identifies the quotient topology of $\ESHM$ as the restriction of the quotient topology of $\SHM$. Again the proof is standard, and included for the reader's convenience.

\begin{proposition}\label{prop:ESHM-subspace-topology}
Under the natural inclusion $\ESHM \subseteq  \SHM$,
the quotient topology on $\ESHM $ coincides with the restriction of the quotient topology of $\SHM$.
\end{proposition}

\begin{proof}
We consider $\SHM$ and $\ESHM$ with the quotient topologies induced by
$q:\SH\to\SHM$ and $r:\ESH\to\ESHM$, respectively, and view $\ESHM$ as a subset
of $\SHM$ via the injective map $j$ from \eqref{eq:commute-quotients}.
We show that $U\subseteq \ESHM$ is open in the quotient topology of $\ESHM$
if and only if there exists an open set $V\subseteq \SHM$ such that
$U = V\cap \ESHM$.

\smallskip\noindent
``Only if'': \
Assume that $U$ is open in $\ESHM$. By definition of the quotient topology.
Then $r^{-1}(U)$ is open in $\ESH$ with respect to the subspace topology inherited from $\SH$.
Hence there exists an open set $O\subseteq \SH$ such that
\begin{equation}\label{eq:O-lift}
r^{-1}(U)=\ESH\cap O.
\end{equation}
Define
\[
W:=\{y\in \SHM:\ \exists \rho\in O \text{ with } q(\rho)=y\}.
\]
Then $W=q(O)$. We claim that $W$ is open in $\SHM$.
Indeed, by Lemma \ref{lem:saturation}
\[
q^{-1}(W)=q^{-1}(q(O)) = O.
\]
Therefore $q^{-1}(W)$ is open in $\SH$, and therefore $W$ is open in $\SHM$ by
Lemma \ref{lem:opennes_of_quotient_TM}

Finally, we show that $U = W\cap \ESHM$.
Taking preimages under $r$ and using $j\circ r = q\circ \iota$, we have
\begin{align*}
r^{-1}(W\cap \ESHM)
& = \{\xi\in\ESH:\ j(r(\xi))\in W\}
\\ & = \{\xi\in\ESH:\ q(\iota(\xi))\in W\}
= \iota^{-1}(q^{-1}(W)).
\end{align*}
Intersecting with $\ESH$ and using \eqref{eq:O-lift} gives
\[
r^{-1}(W\cap \ESHM)=\ESH\cap q^{-1}(q(O))=\ESH\cap O=r^{-1}(U).
\]
Since $r$ is surjective, this implies $W\cap \ESHM = U$.

\smallskip\noindent
``If'': \
Conversely, suppose that $U=V\cap \ESHM$ for some open set $V\subseteq \SHM$.
Then
\begin{align*}
r^{-1}(U)
=r^{-1}(V\cap \ESHM)
& =\{\xi\in\ESH:\ j(r(\xi))\in V\}
\\ & =\{\xi\in\ESH:\ q(\iota(\xi))\in V\}
=\iota^{-1}(q^{-1}(V)).
\end{align*}
Since $q$ is continuous by definition of $\tauMt$, the set $q^{-1}(V)$ is open in $\SH$,
and therefore $\iota^{-1}(q^{-1}(V))$ is open in $\ESH$.
Thus $r^{-1}(U)$ is open in $\ESH$, and by
Lemma \ref{lem:opennes_of_quotient_TM}
$U$ is open in $\ESHM$.
\end{proof}

In what follows, there is no risk of confusion by denoting the quotient topology of $\ESHM$ by $\tauMt$ again, and $\ESHM$ will be assumed to be endowed with this topology.
As a consequence of the preceding results, $\ESHM$ is a Hausdorff space.

In Section \ref{subsec:PVM} we will prove (see Theorem \ref{thm:convex}) that if $\calM$ is an abelian von Neumann algebra, then $$\ESHM = \SHM.$$

\section{Projection-valued measures}\label{sec:examples}

The main result of this section is an intrinsic representation of the Holevo space of the von Neumann algebra generated by a projection-valued measure. As an immediate application, we determine the Holevo space of a finite family of commuting normal operators.

\subsection{Projection-valued measures}\label{subsec:PVM}

Throughout this paragraph, we let
$P:\calF\to\ProjH$ be a projection–valued measure on a given measurable space $(\Omega,\calF)$.

\begin{proposition}[Representation of $\calM$ as $L^\infty(\Omega,\mu)$]
Suppose there exists a $\sigma$-finite measure $\mu$ on $(\Om,\calF)$ with the property that for all $F\in \calF$ we have
$$ \mu(F) = 0 \iff P_F = 0.$$
Then the von Neumann algebra $\calM$ generated by the range of $P$ is $*$–isometric to $L^\infty(\Omega,\mu)$. More precisely, the mapping
$$
\theta: L^\infty(\Omega,\mu) \to \calM,\quad \theta(f) = \int_\Omega f\ud P,
$$
is a $*$–isometry.
\end{proposition}

\begin{proof}
Let $L^\infty(\Omega,P)$ be the Banach space obtained by identifying bounded measurable functions on $\Omega$ when they differ on a set $N\in \calF$ such that $P_N = 0$.
By \cite[Theorem 9.8, Proposition 9.11]{Nee},
the bounded functional calculus for $P$ induces a surjective $*$-isometry
$$
\theta_P: L^\infty(\Omega,P) \to \calM,\quad \theta_P(f) = \int_\Omega f\ud P,
$$
noting that surjectivity follows from the fact that $\theta_P(\one_F) = P_F$ for all $F\in \calF$.
Since $P$ and $\mu$ have the same null sets, the identity mapping $I: f\mapsto f$ is a $*$-isometry from $L^\infty(\Omega,\mu)$ onto $L^\infty(\Omega,P)$. It follows that the mapping $\theta = \theta_P \circ I$ has the stated properties.
\end{proof}

The following proposition allows us to make a special choice of the measure $\mu$.

\begin{proposition}[Cyclic vectors]\label{prop:PVM-cyclic}
There exists a unit vector $h_0\in H$ such that the scalar measure defined by
$$ P_{h_0}(F):=\iprod{P_F h_0}{h_0} , \quad F\in \calF,$$
has the same null sets as $P$.
\end{proposition}

\begin{proof}
A routine application of Zorn's lemma gives us a (possibly uncountable) orthogonal decomposition $ H=\bigoplus_{i\in I} H_i $
into closed nonzero subspaces reducing $P$ (i.e., $P_F(H_i)\subseteq H_i$ for all $F\in \calF$ and $i\in I$)
with the following property: For every $i\in I$ there is a norm--one vector $h_i\in H_i$ which is {\em cyclic} for $P$ restricted to $H_i$,
in the sense the linear span of $\{P_Fh_i:\, F\in\calF\}$ is densely contained in $H_i$.
Since $H$ is separable (recall that this is a standing assumption throughout the paper), any such decomposition is countable. Thus we may fix an orthogonal decomposition $H = \bigoplus_{n\ge 1} H_n $
such that for every $n\ge 1$ there is a norm--one vector $h_n\in H_n$ such that the linear span of $\{P_Fh_n:\, F\in\calF\}$ is densely contained in $H_n$.

We claim that the scalar measure $\mu_n(F) := \iprod{P_Fh_n}{h_n}$ has the property that for every $F\in\calF$,
\[
\mu_n(F)=0 \quad\Longleftrightarrow\quad P_F|_{H_n}=0.
\]
The implication $(\Leftarrow)$ is immediate. For $(\Rightarrow)$, suppose $\mu_n(F)=0$.
Since $P_F$ is an orthogonal projection, $P_F\ge0$ and hence
\[
0=\iprod{P_F h_n}{h_n}=\|P_F h_n\|^2,
\]
so $P_F h_n=0$. For any $G\in\calF$ we then have
\[
P_F(P_G h_n)=P_F P_G h_n = P_{F\cap G}h_n=0,
\]
using the multiplicativity of the PVM on intersections. Thus $P_F$ annihilates the linear span of
$\{P_G h_n:\ G\in\calF\}$, which is dense in $H_n$ by cyclicity of $h_n$ for $P$ on $H_n$.
By continuity, $P_F|_{H_n}=0$.

The unit vector
$$ h_0 := \sum_{n\ge1} 2^{-n/2} h_n.$$
is such that for all $F\in\calF$ we have
$$ P_{h_0}(F) := \iprod{P_Fh_0}{h_0} = \sum_{n\ge1} 2^{-n} \iprod{P_Fh_n}{h_n}, $$
since the cyclic subspaces are mutually orthogonal. Clearly, we have $\iprod{P_Fh_0}{h_0}=0$ if and only if $\iprod{P_Fh_n}{h_n} = 0$ for all $n\ge 1$.
Thus, the scalar measure $P_{h_0}$ satisfies
$$ P_{h_0}(F)=0 \iff\forall n \ \mu_n(F)=0  \iff  \forall n \ P_F|_{H_n}=0  \iff P_F=0. $$
\end{proof}

We now fix a unit vector $h_0\in H$ with the properties stated in the proposition.

\begin{proposition}\label{prop:isom-h0}
Under the above assumptions, the mapping
$$
U:\, \sum_{j=1}^n c_j 1_{F_j}\mapsto \sum_{j=1}^n c_j P_{F_j}h_0
$$
extends uniquely to an isometry from $L^2(\Omega,P_{h_0})$ into $H$.
\end{proposition}

\begin{proof}
First let $f$ be a simple function, say $f=\sum_{j=1}^n c_j 1_{F_j}$,
where $F_j\in\calF$ are pairwise disjoint and $c_j\in\mathbb{C}$. Define
\begin{align}\label{eq:defU}
U(f) := \sum_{j=1}^n c_j P_{F_j}h_0.
\end{align}
Unpacking definitions and using the properties of projection-valued measures,
\begin{align*}
\iprod{U(\one_F)}{U(\one_G)}
& = \iprod{P_F h_0}{P_G h_0} = \iprod{P_{F\cap G}h_0}{h_0}
\\ & = \int_\Om \one_{F\cap G}\ud P_{h_0} = \iprod{\one_F}{\one_G}_{L^2(\Omega,P_{h_0})}.
\end{align*}
By sesquilinearity, this extends to simple functions $f$ and $g$. Hence, $U$ preserves the norm on the simple functions.
Since the simple functions are dense in $L^2(\Omega,P_{h_0})$, $U$  extends uniquely by continuity to an isometry $U$ from $L^2(\Omega,P_{h_0})$ into $H$.
\end{proof}

We write $\mu \, \ll\, P$ if the measure $\mu$ is absolutely continuous with respect to $P$ in the sense that we have $\mu(F)=0$ for every set $F\in\calF$ with $P_F=0$. With this notation, the main theorem of this section reads as follows.

\begin{theorem}[Identification of the Holevo Space]\label{thm:Holevo-for-P}
Let $(\Om,\calF)$ be a measurable space and $P: \calF \to \ProjH$ a projection-valued measure. Let $\calM$ be the von Neumann algebra generated by its range.
Let $h_0\in H$ be any unit vector such that the probability measure $P_{h_0}$ has the same null sets as $P$.
Then $$\ESHM = \SHM,$$ and the identifications
    $$
    [h]_{\calM} \leftrightarrow  \rho P_{h_0}  \leftrightarrow  \rho
    $$
define affine isomorphisms
    \begin{align*}
     \ESHM \ \,& \cong\,
      \Bigl\{\mu \in M_1^+(\Om) :\, \mu \,{\ll}\, P\Bigr\}
 \\ &\cong\,  \Bigl\{\rho\in L^1(\Omega,P_{h_0}) :\, \rho\ge 0,\, \|\rho\|_1=1\,\Bigr\},
    \end{align*}
where $M_1^+(\Om)$ is the convex set of probability measures on $(\Om,\calF)$.
\end{theorem}

By Proposition \ref{prop:PVM-cyclic}, vectors $h_0$ with the stated properties always exist.
Note that the first displayed representation is intrinsic, in the sense that it does not rely on an {\em ad hoc} choice of a vector $h_0$.

\begin{proof}  Since $P_{h_0}$ and $P$ have the same null set and $P_N = 0$ implies $P_h(N) = 0$ for every unit vector $h\in H$, the probability measures $P_h$ are absolutely continuous with respect to $P_{h_0}$.
By the Radon--Nikod\'ym theorem, the first representation therefore follows from the second, with $\rho$ the Radon--Nikod\'ym derivative of $\mu$ with respect to $P_{h_0}$; in the identification $[h]_{\calM} \leftrightarrow  \rho P_{h_0}$ we take $\mu = P_h$. Accordingly, we concentrate on the second representation and will prove that $[h]_{\calM}\mapsto \rho$ sets up an affine isomorphism from $\ESHM $ onto the set
$\{\rho\in L^1(\Omega,P_{h_0}) :\, \rho\ge 0,\, \|\rho\|_1=1\}$.

\smallskip
{\em Step 1} -- By Proposition \ref{prop:PVM-vN-indisc} we have equivalence of pure states $\ket{h} \simM \ket{h'}$ if and only if we have equality of measures $P_h = P_{h'}$. Therefore, $\simM$-equivalent measures share a common Radon--Nikod\'ym density $\rho\in L^1(\Om, P_{h_0})$. This density is nonnegative and satisfies
$$\int_\Omega \rho \ud P_{h_0} = \int_\Om \one \ud P_h = P_h (\Om)  =1.$$

\smallskip
{\em Step 2} -- To see that every $\rho$ in the positive part of the unit sphere of $L^1(\Omega,P_{h_0})$ is obtained in this way, let such a $\rho$ be given. Then $\psi := \rho^{1/2}$ has norm one in $L^2(\Omega,P_{h_0})$. Using the isometry $U$ of Proposition \ref{prop:isom-h0}, let $h := U\psi \in H$. Then
$$\n h\n^2 = \n U\psi\n^2 = \n \psi\n_2^2 = \n \rho\n_1 = 1,$$ and therefore it defines a pure state $\ket{h}$.
To complete the proof we must show that $P_h$ has Radon--Nikod\'ym derivative $\rho = \psi^2$ with respect to $P_{h_0}$, and for this it suffices to check that
\begin{align*} \int_\Om f \ud P_h = \int_\Om f \psi^2 \ud P_{h_0}
\end{align*}
for all $f = \one_F$ with $F\in \calF$.

We claim that for every $F\in\calF$ and every $f\in L^2(\Omega,P_{h_0})$ one has the
intertwining identity
\begin{equation}\label{eq:intertwine-PF-U}
P_F\,U(f)=U(\one_F f).
\end{equation}
Indeed, if $f=\sum_{j=1}^n c_j\one_{F_j}$ is simple, then by \eqref{eq:defU},
\[
P_F U(f)=\sum_{j=1}^n c_j\,P_F P_{F_j}h_0
=\sum_{j=1}^n c_j\,P_{F\cap F_j}h_0
=U\Big(\sum_{j=1}^n c_j\,\one_{F\cap F_j}\Big)
=U(\one_F f).
\]
Now let $f\in L^2(\Omega,P_{h_0})$ and choose simple functions $f_m\to f$ in $L^2(\Om, P_{h_0})$.
Since $P_FU(f_m)=U(\one_F f_m)$ for each $m\ge 1$, we then have
\begin{align*}
\|P_FU(f)-U(\one_F f)\|
& \le \|P_FU(f-f_m)\|+\|U(\one_F(f_m-f))\|
\\ & \le \|f-f_m\|_2+\|\one_F(f_m-f)\|_2\to 0.
\end{align*}
This proves that \eqref{eq:intertwine-PF-U} holds.

Now,
\begin{align*} \int_\Om \one_F \ud P_{h}
& = \iprod{P_F h}{h} = \iprod{P_F U\psi}{U\psi}
\\ & \stackrel{(*)}{=} \iprod{U(\one_F \psi)}{U\psi} \stackrel{(**)}{=}\iprod{\one_F \psi}{\psi} = \int_\Om \one_F \psi^2 \ud P_{h_0},
\end{align*}
where $(*)$ is a consequence of \eqref{eq:intertwine-PF-U} and $(**)$ uses that $U$ is an isometry.

\smallskip{\em Step 3} --
Up to this point, we have shown that the identifications set up bijections. The proof that $\ESHM$ is convex and that the identifications are affine is now routine: if $[\xi_0]_{\calM}$ and $[\xi_1]_{\calM}$ are given elements of $\ESHM$, let $\rho_0$ and $\rho_1$ be the corresponding densities. If $0<\la<1$, then $(1-\la)\rho_0+\la \rho_1$ is also such a density, and hence corresponds to an element $[\xi_\la]_{\calM}$ of $\ESHM$, where $h_\la\in H$ has norm one.
Then, viewing $\xi_0, \xi_1$, and $\xi_\la$ as elements of $\SH$, passing to their equivalence classes results in the identity $[\xi_\la] = (1-\la)[\xi_0]+\la [\xi_1]$ as elements of $\SHM$. This shows that $\ESHM$ is convex as a subset of $\SHM$. Repeating this argument with general convex combinations, it also follows that the bijection is affine.

\smallskip{\em Step 4} -- It remains to prove the equality $\ESHM=\SHM$.
Explicitly, the claim is that the natural inclusion $\ESHM \subseteq \SHM$ induced by passing to quotients in the inclusion $\ESH \subseteq \SH$ is in fact a bijection.

Let $\sigma\in\SH$ be arbitrary and define a probability measure
\[
\mu_\sigma(F):=\tr(P_F\sigma),\qquad F\in\calF.
\]
If $P_F=0$ then $\mu_\sigma(F)=0$, hence $\mu_\sigma\ll P$. Since $P_{h_0}$ has the same
null sets as $P$, also $\mu_\sigma\ll P_{h_0}$. By the Radon--Nikod\'ym theorem there exists
a density $\rho_\sigma\in L^1(\Omega,P_{h_0})$ with $\rho_\sigma\ge0$, $\|\rho_\sigma\|_1=1$ and
\[
\mu_\sigma=\rho_\sigma P_{h_0}.
\]
Let $\psi:=\rho_\sigma^{1/2}\in L^2(\Omega,P_{h_0})$ and set $h:=U\psi\in H$ (so $\|h\|=1$).
By Step~2 we then have $P_h=\rho_\sigma\,P_{h_0}=\mu_\sigma$.

\smallskip
We claim that $\sigma\simM  (h\,\bar\otimes\, h)$. Since $\calM$ is generated by the range of $P$,
it suffices to check equality of expectations on the norm-dense $*$-subalgebra
$\{\theta(s)=\int s\ud P:\ s \text{ simple}\}\subseteq\calM$.
If $s=\sum_{j=1}^n c_j\one_{F_j}$, then
\[
\tr(\theta(s)\sigma)=\sum_{j=1}^n c_j\,\tr(P_{F_j}\sigma)
=\sum_{j=1}^n c_j\,\mu_\sigma(F_j)
=\int_\Omega s\ud  \mu_\sigma,
\]
and similarly,
\[
\iprod{\theta(s)h}{h}
=\sum_{j=1}^n c_j\,\iprod{P_{F_j}h}{h}
=\sum_{j=1}^n c_j\,P_h(F_j)
=\int_\Omega s\ud  P_h
=\int_\Omega s\ud  \mu_\sigma.
\]
Hence $\tr(\theta(s)\sigma)=\iprod{\theta(s)h}{h}$ for all simple $s$.
Now let $f\in L^\infty(\Omega,P_{h_0})$ and choose simple $s_m$ with
$\|f-s_m\|_\infty\to0$. Since $\theta$ is a $*$-isometry,
$\|\theta(f)-\theta(s_m)\|\le \|f-s_m\|_\infty\to0$, and therefore
\[
\tr(\theta(f)\sigma)=\lim_{m\to\infty}\tr(\theta(s_m)\sigma)
=\lim_{m\to\infty}\iprod{\theta(s_m)h}{h}
=\iprod{\theta(f)h}{h}.
\]
Thus $\tr(a\sigma)=\iprod{ah}{h}$ for all $a\in\calM$, i.e.\ $\sigma\simM  h\,\bar\otimes\,h$.

\smallskip
Consequently, every class in $\SHM$ has a pure representative; equivalently $q(\ESH)=\SHM$.
Since the canonical map $\ESH/\!\simM \to \SH/\!\simM $ induced by the inclusion
$\ESH\hookrightarrow\SH$ is always injective, it is therefore bijective, and hence $\ESHM=\SHM$ as convex sets.
\end{proof}

In the case of a finite-dimensional Hilbert space, we have the following result.

\begin{corollary}\label{cor:holevo_finite_dim}
If $H$ is $d$-dimensional and $P : \calF \to \Proj(H)$ is a projection-valued measure,
then $\calF$ is generated by at most $d$ atoms for the measure $P_{h_0}$, the Holevo space takes the form of a simplex over these atoms, and its extreme points are the normalised indicator functions of these atoms.
\end{corollary}
\begin{proof} We use the notation of the preceding theorem.
Let $d = \dim(H)$. Since the cardinality of any family of pairwise orthogonal projections in $H$ is at most $d$, we infer that $\Omega$ is the union of at most $d$ atoms for the measure $P_{h_0}$, say $A_1,\dots,A_k \in \calF$. Then every element of $\{ \varrho \in L^1(\Omega, P_{h_0}): \varrho \geq 0, \norm{\varrho}_1 = 1\}$ is a convex combination of the normalised indicator functions $\one_{A_j}/P_{h_0}(A_j)$, $j=1,\dots,k$, of these atoms, which are its extreme points.
\end{proof}

The following result connects Theorem \ref{thm:Holevo-for-P} to the discussion in the previous section.

\begin{theorem}\label{thm:convex}
     If $\calM$ is abelian, then the inclusion mapping $\ESH\subseteq \SH$ induces an affine isomorphism $$\ESHM \cong \SHM.$$
\end{theorem}
\begin{proof}
    As is well known, (see, e.g., \cite[Proposition 1.21]{Takesaki}), every abelian von Neumann algebra $\calM$ over a separable Hilbert space $H$ is generated by a single self-adjoint element $a$. Applying Theorem \ref{thm:Holevo-for-P} to the projection-valued measure of such a generating selfadjoint element, we obtain the stated result.
\end{proof}

Upon representing the abelian von Neumann algebra $\calM$ as a space $L^\infty(X,\mu)$, every normal state corresponds to a density $\rho\in L^1(X,\mu)$, and as such has an $L^2$ ``square-root'' vector representative
$\psi=\rho^{1/2}$, since $\omega_\rho(f)=\int_X f\rho\ud\mu=\iprod{f\psi}{\psi}$
for all $f\in L^\infty(X,\mu)$. This is consistent with the general standard-form
representations for von Neumann algebras (see, e.g., \cite{Haagerup75}), but here we gave an explicit construction aligned with the operational Holevo-space
viewpoint taken here.

As a final result for this section, we turn to the situation where $\calM$ is the abelian von Neumann algebra generated by the PVMs of multiple, possibly unbounded, normal operators $a_1,\dots,a_n$, assuming these PVMs commute.
If all operators $a_1,\dots,a_n$ are bounded, this is equivalent to assuming that $a_1,\dots,a_n$ commute and $\calM$ is the abelian von Neumann algebra generated by $a_1,\dots,a_n$.

Under the stated assumptions, the joint spectral theorem (see, e.g., \cite[Theorem 4.1]{Landsman}, \cite[Theorem 5.21]{Schm})
guarantees the existence of a unique projection–valued measure
\begin{equation}\label{eq:PVM_commuting_normal_operators}
    P:\calB(\sigma(a_1, \dots,a_n))\to \Proj(H),
\end{equation}
where $\calB$ stands for the Borel $\sigma$-algebra and $\sigma(a_1, \dots, a_n)$ is the joint spectrum,
such that for every bounded Borel function $f\colon X\to\C$ we have
$$
f(a_1,\dots,a_n)= \int_{\sigma(a_1, \dots,a_n)} f\ud P.
$$
In this situation, the von Neumann algebra generated by the projections in the range of $P$ equals $\calM$.

Let $\sigma(a_1,\dots,a_n)$ denote the support of $P$.
Applying the result of the preceding paragraph, we arrive at the following theorem.
When $X$ is a topological Hausdorff space, we denote by $C_{\rm b}(X)$ the Banach space of all bounded continuous functions on $X$.

\begin{corollary}\label{thm:normal}
Let $a_1,\dots,a_n$ be normal operators on $H$ whose PVMs commute, and let
$\calM$ be the von Neumann algebra generated by their spectral measures. Let $h_0\in H$ be any unit vector such that the probability measure $P_{h_0}$ has the same null sets as $P$. Then we have affine isomorphisms
$$\SHM \,\cong\, \ESHM \ \, \cong \,
\Bigl\{ f\in L^1(\sigma(a_1,\dots,a_n), P_{h_0}):\ f\ge 0,\ \|f\|_1=1\Bigr\}.
$$
\end{corollary}
\begin{proof}
    Apply Theorem \ref{thm:Holevo-for-P} to the PVM given by equation \eqref{eq:PVM_commuting_normal_operators}.
\end{proof}

\section{Classicalisation}\label{sec:hidden}

The results proved up to this point can be summarised in the following ``classicalisation'' theorem, although it would be more accurate to say that the theorem provides a way to ``lift quantum observables to classical observables''.

\begin{theorem}[Classicalisation]\label{thm:hidden}
Let $\calM$ be a von Neumann algebra over a Hilbert space $H$. There exist
\begin{itemize}
 \item a Hausdorff topological space $X$
 \item a continuous injective mapping $a\mapsto \wh a$ from $\calM$ to $C_{\rm b}(X)$
\end{itemize}
such that for all $a\in \calM$ and $\ket{h}\in\ESH$ we have
$$ \wh a([h]_{\calM}) = \iprod{ah}{h}.$$
\end{theorem}

\begin{proof}
Take $X:=\ESHM$ endowed with the Hausdorff topology $\tauMt$;
For $a\in\calM$ define $\wh a\in C_{\rm b}(X)$ by
\[
\wh a([h]_{\calM}) := \iprod{ah}{h}.
\]
This is well defined, for if $\ket{h}\simM \ket{h'}$, then
$\iprod{ah}{h}=\iprod{ah'}{h'}$ for all $a\in\calM$. Clearly $\wh a$ is continuous with respect to $\tauMt$.

To see that $a\mapsto \wh a$ is injective, suppose that $a,b\in \calM$ satisfy $\wh a = \wh b$. This means that $\iprod{ah}{h} = \iprod{bh}{h}$ for all norm one vectors $h\in H$, and therefore $a = b$ by a standard polarisation argument.
\end{proof}

This theorem can be viewed as a ``sharp'' version of the hidden variables theorem for unsharp observables of \cite[Section 15.4]{Nee}.
We may state a corresponding version for mixed states, taking $X = \SHM$; in that case, we may view $X$ as Hausdorff space endowed with the topology $\tauMt$.

Classically, the points of a state spaces correspond to the pure states, and by passing to equivalence classes modulo indiscernibility, one factors out the points that cannot be distinguished by the observables at hand. In the quantum context, this suggests to promote the equivalence classes of pure states modulo indiscernibility to be the ``distinguishable points'' of some newly formed space. The theorem shows that the quantum observables can then be lifted to pointwise defined scalar-valued functions on this new space, thus creating a classical model for them.

Let us take a closer look at the case where $\calM$ is the von Neumann algebra generated by a single bounded normal operator $a$ acting on a Hilbert space $H$. Under these assumptions, a classical theorem of von Neumann (see also \cite[Section 9.5]{Nee}) asserts that $\calM$ equals the bounded Borel calculus of $a$,
$$\calM = \bigl\{f(a):\, f\in B_{\rm b}(\sigma(a))\bigr\},$$
where $\sigma(a)$ is the spectrum of $a$. The von Neumann algebra $\calM$ includes (and is in fact generated by) the orthogonal projections in the range of the projection-valued measure $P$ associated with $a$ through the spectral theorem. For any Borel subset $B \subseteq \sigma(a)$, the lifted observable
$\wh{P_B} \in B_{\rm b}(X)$ assigns to $[h]_{\calM}$ the number $\iprod{P_Bh}{h}$. In quantum mechanical language, this number is interpreted as ``the probability of observing an outcome contained in $B$ upon measuring the state $\ket{h}$ using the observable $a$''. In this sense, the lifted classical observable contains all measurement statistics of the quantum observable.

\section{The free particle}\label{sec:free-particle}

As a first (and easy) application of the theory developed up to this point, we determine the Holevo space of the position observables for the free particle in $\R^d$ and show that it can be represented as the set of classical probability densities on $\R^d$. From the abstract results of the preceding section, we know that it takes the form of an abstract simplex on their (joint) spectrum. Rather than applying the abstract theory, we compute the Holevo space from first principles, thus illustrating the scope of the theory by means of a concrete and easily computed instance.

\subsection{The case $d=1$}
Let us first take $d=1$ and let $H = L^2(\R)$ be the Hilbert space of a free particle moving along the real line $\R$.
The position observable for this particle is given by the projection-valued measure
$P:\calB(\R)\to {\rm Proj}(L^2(\R))$, where $\calB(\R)$ is the Borel $\sigma$-algebra of $\R$, defined as
$$ (P_B h)(x) =\one_{B}(x)h(x) $$
for $B\in \calB(\R)$ and $h\in L^2(\R)$.

Suppose now that the pure states $\ket{h}$ and $\ket{h'}$ are indiscernible for $P$, that is,
$$
\int_{\R} \one_B(x)|h(x)|^2 \ud x =\iprod{P_B h}{h} = \iprod{P_B h'}{h'}= \int_{\R} \one_B(x)|h'(x)|^2 \ud x
$$
for all Borel sets $B\subseteq\R$.
This happens if and only if $|h(x)| = |h'(x)|$ for almost all $x\in \R$.
It follows that the Holevo space associated with $P$ is given by
\begin{align}\label{eq:x-meas}
 \mathscr{P}(L^2(\R))/_{\!P} = \Bigl\{ \rho\in L^1(\R):\, \rho\ge 0,\, \int_\R\rho(x)\ud x=1 \Bigr\},
\end{align}
the bijective correspondence being given by
$$ \rho \leftrightarrow |h|^2.$$
This result is a special case of Theorem \ref{thm:Holevo-for-P}, where one may take $h_0$ to be any normalised strictly positive function in $L^2(\R)$, e.g., the density of the Gaussian distribution.

\subsection{The general case}
The same computation can be performed for one or more coordinate position observables for the free particle in $\R^d$; this particle is modelled by the Hilbert space $H = L^2(\R^d)$. In what follows we take $d=3$ for the sake of definiteness.

Let us first consider the position observable along one coordinate axis, say the $x$-coordinate. It is given by the projection-valued measure
$P^{(x)}:\calB(\R)\to {\rm Proj}(L^2(\R^3))$ given by
$$ (P^{(x)}_B h)(x,y,z) =\one_{B}(x)h(x,y,z), \quad B\in \calB(\R), \ h\in L^2(\R^3).$$
We now find that the pure states $\ket{h}$ and $\ket{h'}$ are indiscernible for the von Neumann algebra $\calM_x$ generated by $P^{(x)}$ if and only if for almost all $x\in \R$ we have
$$\int_\R\!\int_\R |h(x,y,z)|^2\ud y\ud z = \int_\R\!\int_\R |h'(x,y,z)|^2\ud y\ud z.$$
Accordingly, the Holevo space for the $x$-observable is the same as before, namely,
\begin{align*}
\mathscr{P}(L^2(\R^3))/_{\!P^{(x)}} = \Bigl\{ \rho\in L^1(\R):\, \rho\ge 0,\, \int_\R \rho(x)\ud x=1 \Bigr\}.
\end{align*}
This time, the correspondence is given by
$$ \rho\leftrightarrow \int_\R\!\int_\R |h(\cdot,y,z)|^2\ud y\ud z .$$
In the same way as before, this result also follows as a special case of Theorem \ref{thm:Holevo-for-P}.

The physics interpretation of this result is that if one is able to observe only the $x$-coordinate of a particle in $\R^3$, then the particle appears to the observer as nonnegative probability densities $\rho$ in one variable.

We could also perform a joint measurement of the
coordinate positions, say the positions along the $x$- and $y$-coordinates. This joint measurement is modelled by the product PVM
\(P^{(x,y)}: \mathscr{B}(\R^2)\to \Proj(L^2(\R^3))\) given by
\[
P^{(x,y)}(B\times C)=P^{(x)}(B)P^{(y)}(C)\qquad (B,C\in\mathcal B(\mathbb R)).
\]
By similar computations as before,
the Holevo space for this PVM is given as
\begin{align}\label{eq:xy-meas}
 \mathscr{P}(L^2(\R^3))/_{\!P^{(x,y)}} =
\Bigl\{ \rho\in L^1(\R^2):\, \rho\ge 0,\,
\int_\R\!\int_\R\ \rho(x,y)\ud x\ud y=1 \Bigr\},
\end{align}
the interpretation being that if one is able to observe only the $x$- and $y$-coordinates of a particle in $\R^3$, then the particle appears to the observer as nonnegative probability densities $\rho$ in two variables.

One should carefully distinguish between indiscernibility for $P^{(x,y)}$ and for the pair
$\{P^{(x)},P^{(y)}\}$. The pure states $\ket{h}$ and $\ket{h'}$ are indiscernible for $\{P^{(x)},P^{(y)}\}$ when they have the same $x$- and $y$-marginals, i.e.,
\[
\int_{\R}\!\int_{\R}|h(x,y,z)|^2\ud  y\ud  z=\int_{\R}\!\int_{\R}|h'(x,y,z)|^2\ud y\ud z\quad\text{ a.a.\ }x\in\R,
\]
and
\[
\int_{\R}\!\int_{\R}|h(x,y,z)|^2\ud  x\ud  z=\int_{\R}\!\int_{\R}|h'(x,y,z)|^2\ud x\ud z\quad\text{ a.a.\ }y\in\R.
\]
This condition forgets correlations between $x$ and $y$.
The corresponding Holevo space is naturally identified with
\[
\mathscr{P}(L^2(\R^3))/_{\!\{P^{(x)},P^{(y)}\}}
=\Bigl\{(\rho_x,\rho_y):\, \rho_x,\rho_y\in L^1(\R),\ \rho_x,\rho_y\ge 0,\
\|\rho_x\|_1=\|\rho_y\|_1=1\Bigr\}.
\]
This identification is surjective: given $(\rho_x,\rho_y)\in
\mathscr{P}(L^2(\R^3))/_{\!\{P^{(x)},P^{(y)}\}}$,
choose $\phi,\psi,\chi\in L^2(\R)$ with $|\phi|^2=\rho_x$, $|\psi|^2=\rho_y$, and $\|\chi\|_2=1$,
and consider the function $h(x,y,z):=\phi(x)\psi(y)\chi(z)$.

A pair of states $\ket{h}$ and $\ket{h'}$ that is discerned by $\sim_{P^{(x,y)}}$ but not by  $\sim_{\{P^{(x)},P^{(y)}\}}$ can be constructed as follows.
Choose disjoint Borel sets $A_1,A_2,B_1,B_2\subseteq\R$ with
$0<|A_i|,|B_j|<\infty$ and fix $\chi\in L^2(\R)$ with $\|\chi\|_2=1$. Define
the ``Bell-type'' states
\begin{align*}
h(x,y,z) & := \frac1{\sqrt2}\Bigl(\frac{\one_{A_1}(x)\one_{B_1}(y)}{|A_1|^{1/2}|B_1|^{1/2}}
+\frac{\one_{A_2}(x)\one_{B_2}(y)}{|A_2|^{1/2}|B_2|^{1/2}}\Bigr)\chi(z),
\\
h'(x,y,z) & := \frac1{\sqrt2}\Bigl(\frac{\one_{A_1}(x)\one_{B_2}(y)}{|A_1|^{1/2}|B_2|^{1/2}}
+\frac{\one_{A_2}(x)\one_{B_1}(y)}{|A_2|^{1/2}|B_1|^{1/2}}\Bigr)\chi(z).
\end{align*}
Then $h\sim_{\{P^{(x)},P^{(y)}\}}h'$, but
\[
\iprod{P^{(x)}(A_1)P^{(y)}(B_1)h}{h}=\frac12,
\qquad
\iprod{P^{(x)}(A_1)P^{(y)}(B_1)h'}{h'}=0,
\]
so $h\not\sim_{\!P^{(x,y)}}h'$.

The Holevo spaces for the (joint or separate) $x$-, $y$-, and $z$-observables can be treated in the same way.

\begin{remark}
Using the canonical isometric identification
\[
L^{2}(\R^{2})\cong L^{2}(\R)\otimes L^{2}(\R)
\]
(with slight abuse of notation, the right-hand side denotes the Hilbert space completion of the algebraic tensor product), the sharp joint position PVM $P^{(x,y)}$ is analogous to Charlie's joint polarisation measurement in the EPR experiment discussed in Section \ref{subsec:Holevo-EPR}, while the pair $\{P^{(x)},P^{(y)}\}$ of marginal PVMs is analogous to Alice's and Bob's separate measurements in this experiment.

This analogy is, however, purely mathematical. In an actual EPR/Bell experiment the tensor product
decomposition $H=\C^2\otimes \C^2$ is fixed by physically distinguished subsystems, which selects the local observable algebras
$\calL(\C^2)\otimes I$ and $I\otimes\calL(\C^2)$. For a single particle in $\R^{2}$, interpreting the two tensor factors in the tensor product $L^{2}(\R^{2})\cong L^{2}(\R)\otimes L^{2}(\R)$ as ``Alice'' and ``Bob'' would also require additional physical structure.
\end{remark}

\section{The EPR and Bell experiments}\label{sec:EPR_and_Bell}

In this section we show how Holevo spaces provide an explicit way to track the
\emph{relational} state space in multi-observer scenarios, culminating in the Bell experiment. We proceed step by step: we first analyse the single-qubit case, then
move to the two-qubit setting of the EPR experiment, and finally incorporate variable polariser settings as in Bell tests. This leads to a concrete geometric
description that makes two limitations transparent: (a) from their local data, Alice and Bob cannot reconstruct Charlie's joint description for a fixed choice
of settings, and (b) from Charlie's perspective there is, in general, no single `global' Holevo model that simultaneously captures all polariser settings at once.

\subsection{The qubit}\label{subsec:Holevo-qubit}

As a preparation for the next two paragraphs, in this paragraph we compute the Holevo space of the qubit (see also \cite[Section 15.4]{Nee}).

\subsubsection{Basic set-up}
Mathematically, the qubit is modelled on the Hilbert space $H = \C^2$.
Keeping in mind that norm-one vectors define the same pure state if they differ by a multiplicative scalar of
modulus one, in terms of the standard basis $\{\ket{0},\ket{1}\}$ of $\C^2$, every pure state can be written as
\begin{align*}\cos (\theta/2) \ket{0}+e^{i\eta}\sin (\theta/2) \ket{1}, \quad \theta\in [0,\pi], \ e^{i\eta} \in \T,
\end{align*}
where $\T = \S^1$ is the unit circle in the complex plane. In this representation, the angle $\theta\in [0,\pi]$ is unique; moreover,
for $\theta\in (0,\pi)$ the angle $\eta$ is unique up to a multiple of $2\pi$;
for $\theta \in\{0,\pi\}$ the angle $\eta$ can be chosen arbitrarily.
In this polar coordinate description, the pure state space $\mathscr{P}(\C^2)$ is represented as the unit sphere $\S^2$ in $\R^3$, the so-called {\em Bloch sphere}.

Every selfadjoint operator on $\C^2$ admits a unique decomposition in the basis ${I,\sigma_1,\sigma_2,\sigma_3}$, where $I$ is the identity matrix and
$$\sigma_1 = \begin{pmatrix}\, 0 & \, 1\, \\ 1 & 0  \end{pmatrix}, \quad
  \sigma_2 = \begin{pmatrix}\, 0 & \,-i\, \\ i & 0  \end{pmatrix}, \quad
  \sigma_3 = \begin{pmatrix}\, 1 & \, 0\, \\ 0 & -1 \end{pmatrix}$$
are the {\em Pauli matrices}, which correspond to the spin observables associated with polarisation measurements along the three coordinate axes. Each of the observables
$\sigma_1$, $\sigma_2$, $\sigma_3$ has spectrum $\{-1,1\}$. In the physics literature, these observables are frequently referred to as the {\em Pauli-$x$, Pauli-$y$, and Pauli-$z$ observables}.

In the discussion below, we will specifically work with $\sigma_3$, which has the advantage that the computational basis is also an orthonormal eigenbasis
for the $\{+1,-1\}$-valued observable $\sigma_3$.
To simply future notation, we write the outcome space $\{+1,-1\}$ as $\{+,-\}$.
In this notation, the PVM $P_3$ of $\sigma_3$ is uniquely given on the singletons of $\{+,-\}$ by
$$ P_3(\{+\}) = \ket{0}\bra{0}, \quad P_3(\{-\}) = \ket{1}\bra{1}.$$

\subsubsection{Bloch sphere description of the Holevo space}

Using the Bloch sphere description for pure states, we will determine when two pure states
\begin{align*}
\ket{h} = \cos (\theta/2) \ket{0} + e^{i\eta }\sin(\theta/2)  \ket{1} \ \hbox{ and } \  \ket{h'} = \cos (\theta'/2) \ket{0} + e^{i\eta' }\sin(\theta'/2)  \ket{1}
\end{align*}
are indiscernible for the Pauli-$z$ observable $\sigma_3$, or equivalently, for the PVM $P_3$ of $\sigma_3$.
According to the observation in Example \ref{ex:spin},
this happens if and only if
\begin{align*}
\iprod{\sigma_3 h}{h} = \iprod{\sigma_3 h'}{h'}.
\end{align*}
A direct computation gives
\begin{align*}
\iprod{ \sigma_3 h}{h}
=\cos\theta
\ \ \hbox{ and } \ \ \iprod{\sigma_3 h'}{h'}=\cos\theta'.
\end{align*}
Hence
\[
\ket h\sim_{P_3}\ket{h'}
\quad\Longleftrightarrow\quad
\cos\theta=\cos\theta'.
\]
Since $\theta,\theta'\in[0,\pi]$ and $\cos$ is injective on this interval, this is equivalent to $\theta=\theta'$.

This represents the Holevo space of $P_3$ as
\begin{align}\label{eq:Holevo-qubit-polar}\mathscr{P}(\C^2)/_{P_3} \ \cong \ [0,\pi].
\end{align}
For each $\theta\in (0,\pi)$, the corresponding equivalence class is
the circles parallel to the equator at angle $\theta$;
for $\theta = 0$ and $\theta=\pi$ these circles degenerate to a point (the north and south poles, respectively).

\subsubsection{Intrinsic description of the Holevo space}

An intrinsic ``spectral'' description of the Holevo space for $P_3$ can be given in line with Corollary \ref{cor:holevo_finite_dim}. According to this theorem, the Holevo space is the space of probability distributions on the outcome space $\{+,-\}$; every such measure is (trivially) absolutely continuous with respect to $P_3$.
This leads to the description of the Holevo space as
\begin{align}\label{eq:Holevo-qubit}\mathscr{P}(\C^2)/_{P_3} \ \cong \ \Delta^1,
\end{align}
where
$$ \Delta^1 :=\Bigl\{p = (p_+,p_-)\in[0,1]^2:\ p_{+}+p_{-}=1\Bigr\}.$$
This can be shown from first principles as follows.
Let $\ket{h}=c_0 \ket{0}+c_1 \ket{1}$ with $|c_0|^2+|c_1|^2=1$ be a pure state.
If both $c_0\neq 0$ and $c_1\neq 0$, the indiscernibility class consists of all vectors obtained
by varying the relative phase:
\[
[h]_{P_3}
=\Bigl\{c_0 \ket{0}+\zeta c_1 \ket{1}:\ \zeta\in\T\Bigr\}.
\]
In the cases $c_0=0$ and $c_1=0$ the equivalence class collapses to a singleton: if $c_0=0$, then $[h]_{P_3} =\{\ket{1}\}$; if $c_1=0$, then $[h]_{P_3} = \{ \ket{0}\}$.

Under the identification $\mathscr{P}(\C^2)/_{P_3}\cong\Delta^1$,
the Holevo lifts of the two projections are the coordinate maps:
\[
\widehat{P_3(\{+ \})}(p)= p_+,  \quad  \widehat{P_3(\{-\})}(p)= p_-,
\] for $p = (p_+, p_-) \in \Delta^1$.

The Bloch sphere description \eqref{eq:Holevo-qubit-polar} is recovered from the intrinsic description \eqref{eq:Holevo-qubit} by writing $p\in \Delta^1$ as
$$p = (p_+,p_-) = (\cos^2(\theta/2),\, \sin^2(\theta/2))$$
for a unique $\theta\in [0,\pi]$.
Under this rewriting, the above Holevo lifts take the form of Bell's hidden variables for the qubit \cite{Bell1966}; see also \cite{Hol} and \cite[Section 15.4]{Nee}.

\subsection{The EPR experiment}\label{subsec:Holevo-EPR}

Next we turn to correlation measurements in a two-qubit system.

\subsubsection{Basic set-up} The two-qubit system corresponds to the choice $H = \C^2\otimes \C^2$.

We consider the standard \emph{orthonormal} computational basis of $\C^2\otimes\C^2$ given by
\[
\ket{00}= \begin{pmatrix}1\\0\end{pmatrix}\otimes\begin{pmatrix}1\\0\end{pmatrix},\ \
\ket{01}= \begin{pmatrix}1\\0\end{pmatrix}\otimes\begin{pmatrix}0\\1\end{pmatrix},\ \
\ket{10}= \begin{pmatrix}0\\1\end{pmatrix}\otimes\begin{pmatrix}1\\0\end{pmatrix},\ \
\ket{11}= \begin{pmatrix}0\\1\end{pmatrix}\otimes\begin{pmatrix}0\\1\end{pmatrix}.
\]

In the EPR experiment \cite{EPR}, Alice and Bob perform polarisation measurements on their respective subsystems.
We model these by the PVMs $P_{3;A}$ and $P_{3;B}$ on $\{+,-\}$
associated with $\sigma_3\otimes I$ and $I\otimes\sigma_3$, respectively, uniquely determined by
\[
P_{3;A}(\{\eps\})=P_3(\{\eps\})\otimes I,\qquad
P_{3;B}(\{\eps\})=I\otimes P_3(\{\eps\}),\qquad \eps\in\{+,-\}.
\]
An external observer Charlie can record the \emph{pair} of outcomes. This is described by the product PVM
$P_{3;C}$ on $\{+,-\}\times\{+,-\}$, uniquely determined by
\begin{equation*}
P_{3;C}(\{(\eps_A,\eps_B)\})
= P_3(\{\eps_A\})\otimes P_3(\{\eps_B\}),\qquad (\eps_A,\eps_B)\in\{+,-\}\times\{+,-\} .
\end{equation*}

\begin{remark}
One should not confuse the PVM $P_{3;C}$ on $\{+,-\}\times\{+,-\}$
with the PVM on $\{+,-\}$ of $\sigma_3\otimes\sigma_3$.
\end{remark}

We are interested in indiscernibility of pure states $\ket{h}$ and $\ket{h'}$ with respect to $P_{3;C}$, i.e.
\[
\iprod{P_{3;C}(F)h}{h}=\iprod{P_{3;C}(F)h'}{h'}
\qquad\text{for all }F\subseteq\{+,-\}\times\{+,-\}.
\]
Equivalently, it suffices to check this for the four singletons $F=\{(\eps_A,\eps_B)\}$.
If $\ket{h}\sim_{P_{3;C}}\ket{h'}$, then in particular $\ket{h}$ and $\ket{h'}$ have the same marginals for $P_{3;A}$ and $P_{3;B}$.
The converse does not hold: equality of the marginals does not determine the joint distribution.

\subsubsection{The Holevo space for $P_{3; C}$}\label{subsec:C-EPR}
Let us start by determining a representation of the Holevo space for the observable $P_{3;C}$. As in the case of the single qubit, a ``Bloch sphere type'' description can be given in terms of angles, but the details are cumbersome.
In contrast, the intrinsic ``spectral'' description offered by Theorem \ref{thm:Holevo-for-P} is immediate, and can again be described explicitly without much effort.

Theorem \ref{thm:Holevo-for-P} represents the Holevo space for $P_{3;C}$ as the set of probability measures
on the outcome space $\{+,-\}\times \{+,-\}$; again every such measure is (trivially) absolutely continuous with respect to $P_{3;C}$.
This leads to the description of the Holevo space as the simplex
$$ \Delta^3 :=\Bigl\{p\in[0,1]^4:\ p_{++} \,+\, p_{+-} \,+\, p_{-+} \,+\, p_{--} \,=\,1\Bigr\}.$$
The bijective correspondence
\begin{align}\label{eq:Holevo-3C} \mathscr{P}(\C^2\otimes\C^2)/_{P_{3;C}}  \ \cong\ \Delta^3
\end{align}
can be made explicit as follows. For every pure state $\ket{h}$ define
\[
\mu_h(\eps_A,\eps_B)
:=\iprod{P_{3;C}(\{(\eps_A,\eps_B)\})h}{h}, \qquad (\eps_A,\eps_B) \in \{+,-\}\times\{+,-\}.
\]
Then $\mu_h = (\mu_h(++),\, \mu_h(+-),\, \mu_h(-+),\, \mu_h(--)) \in \Delta^3$, and by the same argument as in the qubit case, the mapping
$ [h]_{P_{3;C}} \mapsto \mu_h$ sets up the bijective correspondence \eqref{eq:Holevo-3C}.

In this representation, similar calculations as in the single-qubit case show that we can describe the equivalence classes generically as
$\T\times \T \times \T$ for states in the interior of $\Delta^3$,
with degenerations to $\T\times \T$ or $\T$ or a singleton depending on whether one, two, or three numbers $p_{\pm,\pm}$ vanish.

The Holevo lift of the
projections $P_{3;C}(\{(\eps_A,\eps_B)\})$ are the coordinate functions:
\[
\widehat{P_{3;C}(\{(\eps_A,\eps_B)\})}(p)=p_{\eps_A,\eps_B},
\qquad (\eps_A,\eps_B)\in\{+,-\}\times \{+,-\},
\]
for $p = (p_{++},\, p_{+-},\, p_{-+},\, p_{--} )\in \Delta^3$.

\subsubsection{The Holevo space for the pair $\{P_{3;A},P_{3;B}\}$}\label{subsec:AB-EPR}

Let $P_{3;A}$ and $P_{3;B}$ be the $\{+,-\}$-valued PVMs associated with
$\sigma_3\otimes I$ and $I\otimes\sigma_3$, uniquely defined on the singletons by
\[
P_{3;A}(\{\eps\}) := P_3(\{\eps\})\otimes I,\qquad
P_{3;B}(\{\eps\}) := I\otimes P_3(\{\eps\}),\qquad \eps\in\{+,-\},
\]
with $P_3$ the PVM of $\sigma_3$.
We recall that pure states $\ket{h}$ and $\ket{h'}$ are said to be \emph{indiscernible}
for the pair $\{P_{3;A},P_{3;B}\}$, notation $\ket{h}\sim_{\{P_{3;A}, P_{3;B}\}}\ket{h'}$,
if their marginal outcome measures coincide, i.e.\ for all $F\subseteq \{+,-\}$,
\[
\iprod{P_{3;A}(F)h}{h}=\iprod{P_{3;A}(F)h'}{h'}
\quad\text{and}\quad
\iprod{P_{3;B}(F)h}{h}=\iprod{P_{3;B}(F)h'}{h'}.
\]
The equivalence class of $\ket{h}$ under $\sim_{{\{P_{3;A}, P_{3;B}\}}}$ is denoted by $[h]_{{\{P_{3;A}, P_{3;B}\}}}$.
Importantly, we do \emph{not} pass to the von Neumann algebra generated by $P_{3;A}$ and $P_{3;B}$ (as shown in Example~\ref{ex:counterexample_prop_indisc_specmeas} this would alter the equivalence relation).

For $\ket{h}\in\calP(C^2\otimes \C^2)$ and $\eps\in\{+,-\}$ define the marginal probabilities
\[
\mu_{h,A}(\eps):=\iprod{P_{3;A}(\{\eps\})h}{h},\qquad
\mu_{h,B}(\eps):=\iprod{P_{3;B}(\{\eps\})h}{h}.
\]
Then $\mu_{h,A} := (\mu_{h,A}(+), \mu_{h,A}(-))\in\Delta^1$ and $\mu_{h,B} := (\mu_{h,B}(+), \mu_{h,B}(-))\in\Delta^1$ depend only on the
$\sim_{\{P_{3;A},P_{3;B}\}}$--equivalence class of $\ket{h}$. Hence, the mapping
\begin{align}\label{eq:surj}
[h]_{\{P_{3;A},P_{3;B}\}}\mapsto (\mu_{h,A},\mu_{h,B})
\end{align}
is well defined and injective from $\calP(C^2\otimes\C^2)/_{\!{\{P_{3;A},P_{3;B}\}}}$ to $\Delta^1\times\Delta^1$.
We claim that this mapping is also surjective. Fix an arbitrary point $(\mu_A,\mu_B)\in\Delta^1\times\Delta^1$ and write
$\mu_A=(\mu_A(+),\mu_A(-))$ and $\mu_B=(\mu_B(+),\mu_B(-))$.
Set
\[
\ket{\phi_A}:=\mu_A(+)^{1/2}\ket{0}+\mu_A(-)^{1/2}\ket{1},\qquad
\ket{\phi_B}:=\mu_B(+)^{1/2}\ket{0}+\mu_B(-)^{1/2}\ket{1},
\]
and define $\ket{h}:=\ket{\phi_A}\otimes\ket{\phi_B}$.
Then $\|h\|=1$ and
$$\mu_{h,A}(\eps)=\mu_A(\eps),\quad  \mu_{h,B}(\eps)=\mu_B(\eps).$$ This proves surjectivity of the mapping \eqref{eq:surj}. It follows that this mapping establishes a bijective correspondence
\[
\calP(\C^2\otimes\C^2)/_{\!\{P_{3;A},P_{3;B}\}}\ \cong\ \Delta^1\times\Delta^1.
\]

\subsubsection{Charlie versus Alice \& Bob}\label{subsec:interpreting-Holevo}

Given the state $$\ket{h} \,=\, c_{00}\ket{00} \,+\, c_{01}\ket{01} \,+\, c_{10}\ket{10} \,+\, c_{11}\ket{11}$$ with $|c_{00}|^2 \,+\, |c_{01}|^2 \,+\, |c_{10}|^2 \,+\, |c_{11}|^2=1$,
Alice and Bob can only access their respective marginals, and thus they have access
to the point
\begin{align}\label{eq:holevo-point-AB}
\Bigl(
\bigl(|c_{00}|^2+|c_{01}|^2,\,|c_{10}|^2+|c_{11}|^2\bigr),\,
      \bigl(|c_{00}|^2+|c_{10}|^2,\,|c_{01}|^2+|c_{11}|^2\bigr)
\Bigr)
\end{align}
in $\Delta^1\times\Delta^1$. In this sense $\Delta^1\times\Delta^1$ encodes what Alice and Bob can  ``infer locally'':
Alice can infer the point $(|c_{00}|^2+|c_{01}|^2,\,|c_{10}|^2+|c_{11}|^2)\in \Delta^1$ and Bob can infer the point $(|c_{00}|^2+|c_{10}|^2,\,|c_{01}|^2+|c_{11}|^2)\in \Delta^1.$
Charlie, on the other hand, has access to the point
\begin{align}\label{eq:holevo-point-C}
(|c_{00}|^2,\,|c_{01}|^2,\,|c_{10}|^2,\,|c_{11}|^2)
\end{align}
of $ \Delta^3$.
The ``forgetful mapping'' from $\Delta^3$ to $\Delta^1\times\Delta^1$ is the passage from \eqref{eq:holevo-point-C} to \eqref{eq:holevo-point-AB},
\begin{align*}
\ & \bigl(|c_{00}|^2,|c_{01}|^2,|c_{10}|^2,|c_{11}|^2\bigr)
\\ & \qquad \mapsto\,\Bigl((|c_{00}|^2+|c_{01}|^2,\,|c_{10}|^2+|c_{11}|^2),\ (|c_{00}|^2+|c_{10}|^2,\,|c_{01}|^2+|c_{11}|^2)\Bigr)
\end{align*} which ``forgets'' the correlations.

For example, the equivalence class of the Bell state $h_{\rm Bell}$
corresponds to the point $(\frac12,0,0,\frac12)\in \Delta^3$ for Charlie, and
to the point $((\frac12,\frac12), (\frac12,\frac12))\in\Delta^1\times\Delta^1$,
confirming that Charlie sees the correlations, while Alice and Bob each see a
fair coin locally.

\subsection{The Bell experiment}\label{subsec:Holevo-Bell}

The Bell experiment modifies the EPR experiment by introducing randomised polariser settings on either side of the experiment before performing the polarisation measurements.
Alice and Bob set the angles of their polarisers independently.

\subsubsection{Basic Set-up}
Consider the rotation matrices
\[
R^\gamma :=
\begin{pmatrix}
\cos\gamma & -\sin\gamma \\
\sin\gamma & \cos\gamma
\end{pmatrix},
\qquad (R^\gamma)^*=R^{-\gamma}
=
\begin{pmatrix}
\cos\gamma & \sin\gamma \\
-\sin\gamma & \cos\gamma
\end{pmatrix}.
\]
Fixing angles $\gamma_A,\gamma_B\in\R$, we also consider the unitary operator
\[
U^{\gamma_A,\gamma_B}:=R^{\gamma_A}\otimes R^{\gamma_B}.
\]

The Bell analogue of $P_{3;C}$ for the fixed setting $(\gamma_A,\gamma_B)$ is the PVM
$P_{3;C}^{\gamma_A,\gamma_B}$ on $\{+,-\}\times\{+,-\}$, uniquely determined by
\begin{equation*}
P_{3;C}^{\gamma_A,\gamma_B}(\{(\eps_A,\eps_B)\})
 := \bigl((R^{\gamma_A})^*P_3(\{\eps_A\})R^{\gamma_A}\bigr)\otimes
 \bigl((R^{\gamma_B})^*P_3(\{\eps_B\})R^{\gamma_B}\bigr).
\end{equation*}
We have
\begin{align}\label{eq:def-P3Cgamma}
P_{3;C}^{\gamma_A,\gamma_B}(F)=(U^{\gamma_A,\gamma_B})^*\,P_{3;C}(F)\,U^{\gamma_A,\gamma_B},
\qquad F\subseteq\{+,-\}^2.
\end{align}
An easy computation shows that for the Bell state
$\ket{h_{\rm Bell}}=\frac1{\sqrt2}(\ket{00}+\ket{11})$ we have
\begin{equation}\label{eq:P3C}
\begin{aligned}
\iprod{P_{3;C}^{\gamma_A,\gamma_B}(\{(+,+)\})h_{\rm Bell}}{h_{\rm Bell}}
&=\frac12\cos^2(\gamma_A-\gamma_B),\\
\iprod{P_{3;C}^{\gamma_A,\gamma_B}(\{(+,-)\})h_{\rm Bell}}{h_{\rm Bell}}
&=\frac12\sin^2(\gamma_A-\gamma_B),\\
\iprod{P_{3;C}^{\gamma_A,\gamma_B}(\{(-,+)\})h_{\rm Bell}}{h_{\rm Bell}}
&=\frac12\sin^2(\gamma_A-\gamma_B),\\
\iprod{P_{3;C}^{\gamma_A,\gamma_B}(\{(-,-)\})h_{\rm Bell}}{h_{\rm Bell}}
&=\frac12\cos^2(\gamma_A-\gamma_B).
\end{aligned}
\end{equation}

\subsubsection{The Holevo space for $P_{3;C}^{\gamma_A,\gamma_B}$ }
Let $\calM^{\gamma_A,\gamma_B}$ be the von Neumann algebra generated by the range of
$P_{3;C}^{\gamma_A,\gamma_B}$, and set $\calM:=\calM^{0,0}$.
Since $P_{3;C}^{\gamma_A,\gamma_B}=(P_{3;C})^{U^{\gamma_A,\gamma_B}}$, Theorem~\ref{thm:Holevo-invariance-under-rotation}
gives:

\begin{proposition}\label{prop:Bell-Holevo-same-as-EPR}
The mapping
\[
[h]_{P_{3;C}^{\gamma_A,\gamma_B}}\longmapsto [U^{\gamma_A,\gamma_B}h]_{P_{3;C}}
\]
is a well-defined bijection
\[
\mathscr P(\C^2\otimes\C^2)/\!\sim_{P_{3;C}^{\gamma_A,\gamma_B}}
\ \cong\
\mathscr P(\C^2\otimes\C^2)/\!\sim_{P_{3;C}}.
\]
\end{proposition}

\subsubsection{The Holevo space for $\{P_{3;A}^{\gamma_A}, P_{3;B}^{\gamma_B}\}$}

We now show that the Alice/Bob Holevo space for $\sim_{\{P_{3;A}^{\gamma_A}, P_{3;B}^{\gamma_B}\}}$ is independent of the polariser settings.

For a fixed choice of angles \(\gamma_A,\gamma_B\), let the rotated local PVMs on \(\{+,-\}\) be given by
\[
P_{3;A}^{\gamma_A}(\{\eps\})
:=\bigl((R^{\gamma_A})^*P_3(\{\eps\})R^{\gamma_A}\bigr)\otimes I,\quad
P_{3;B}^{\gamma_B}(\{\eps\})
:=I\otimes\bigl((R^{\gamma_B})^*P_3(\{\eps\})R^{\gamma_B}\bigr),
\]
and define \(\ket{h}\sim_{\{P_{3;A}^{\gamma_A} , P_{3;B}^{\gamma_B}\}}\ket{h'}\) to mean equality of the two marginals for
\(P_{3;A}^{\gamma_A}\) and \(P_{3;B}^{\gamma_B}\).
We claim that
\begin{align}\label{eq:Holevo-AB}\mathscr P(\C^2\otimes\C^2)/\!{}_{\{P_{3;A}^{\gamma_A} , P_{3;B}^{\gamma_B}\}} \ \cong \ \Delta^1\times\Delta^1
\end{align}
via the mapping
\[
[h]_{\{P_{3;A}^{\gamma_A} , P_{3;B}^{\gamma_B}\}}
\mapsto \bigl(\mu^{\gamma_A}_{h,A},\mu^{\gamma_B}_{h,B}\bigr),
\]
where
\[
\mu^{\gamma_A}_{h,A}(\eps):=\iprod{P_{3;A}^{\gamma_A}(\{\eps\})h}{h},\quad
\mu^{\gamma_B}_{h,B}(\eps):=\iprod{P_{3;B}^{\gamma_B}(\{\eps\})h}{h}.
\]
It is clear that this mapping is well defined and injective.
To see that it is surjective, fix $(\nu_A,\nu_B)\in\Delta^1\times\Delta^1$, with
$\nu_A=(\nu_A(+),\nu_A(-))$ and $\nu_B=(\nu_B(+),\nu_B(-))$.
Define single-qubit unit vectors
\begin{align*}
\ket{\phi_A}
&= \nu_A(+)^{1/2}\ket{0_{\gamma_A}}+\nu_A(-)^{1/2}\ket{1_{\gamma_A}}, \\
\ket{\phi_B}
&= \nu_B(+)^{1/2}\ket{0_{\gamma_B}}+\nu_B(-)^{1/2}\ket{1_{\gamma_B}},
\end{align*}
where for $\eta \in \{0,1\}$ $\ket{\eta_{\gamma}}:=(R^{\gamma})^*\ket{\eta}$, and set
\[
\ket{h}:=\ket{\phi_A}\otimes\ket{\phi_B}.
\]
Then $\mu^{\gamma_A}_{h,A}=\nu_A$ and $\mu^{\gamma_B}_{h,B}=\nu_B$, proving surjectivity. This proves \eqref{eq:Holevo-AB}.

\smallskip

In the unrotated setting,  under the identification
\[\mathscr P(\C^2\otimes\C^2)/\!{}_{\{P_{3;A},P_{3;B}\}}
\ \cong\ \Delta^1\times\Delta^1\] given by \[
[h]_{\{P_{3;A},P_{3;B}\}}\mapsto (\mu_{h,A},\mu_{h,B}),\]
where
\(\mu_{h,A}(\eps) = \iprod{P_{3;A}(\eps)h}{h} \) and \( \mu_{h,B}(\eps) = \iprod{P_{3;B}(\eps)h}{h}\),
the equivalence class of the Bell state
$$\ket{h_{\rm Bell}}:=\frac1{\sqrt2}(\ket{00}+\ket{11})$$
corresponds to the point
$((\frac12,\frac12),(\frac12,\frac12)) \in \Delta^1\times\Delta^1$.
The next theorem is the stronger statement that the equivalence class of
\(\ket{h_{\rm Bell}}\) always maps to the same point $ ((\frac12,\frac12),(\frac12,\frac12)) \in \Delta^1\times\Delta^1$,
independently of \((\gamma_A,\gamma_B)\).
Thus, the local $P_{3;A},P_{3;B}$-statistics of $\ket{h_{\rm Bell}}$ are independent of
$(\gamma_A,\gamma_B)$ (even though Charlie's joint statistics do depend on the setting).

\begin{theorem}[The rotated Bell state lies in the same equivalence class]\label{thm:hide-2-intrinsic}
Set
$U^{\gamma_A,\gamma_B}:=R^{\gamma_A}\otimes R^{\gamma_B}$.
Then for all angles $\gamma_A,\gamma_B$ one has equality of equivalence classes
\[
\bigl[(R^{\gamma_A}\otimes R^{\gamma_B}) h_{\rm Bell}\bigr]_{{\{P_{3;A}, P_{3;B}\}}}
=
\bigl[h_{\rm Bell}\bigr]_{{\{P_{3;A}, P_{3;B}\}}}.
\]
Equivalently, $\ket{(R^{\gamma_A}\otimes R^{\gamma_B})h_{\rm Bell}}$ and $\ket{h_{\rm Bell}}$ have the same
marginals for $P_{3;A}$ and $P_{3;B}$.
\end{theorem}

\begin{proof}
For $\eps\in\{+,-\}$ we have, with $U=R^{\gamma_A}\otimes R^{\gamma_B}$ as before,
\begin{align*}
\iprod{P_{3;A}(\{\eps\})\,Uh_{\rm Bell}}{Uh_{\rm Bell}}
& =\iprod{U^*(P_{3;A}(\{\eps\}))U\,h_{\rm Bell}}{h_{\rm Bell}}
=\iprod{P_{3;A}^{\gamma_A}(\{\eps\})\,h_{\rm Bell}}{h_{\rm Bell}}, \\
\iprod{P_{3;B}(\{\eps\})\,Uh_{\rm Bell}}{Uh_{\rm Bell}}
& =\iprod{U^*(P_{3;B}(\{\eps\}))U\,h_{\rm Bell}}{h_{\rm Bell}}
=\iprod{P_{3;B}^{\gamma_B}(\{\eps\})\,h_{\rm Bell}}{h_{\rm Bell}}.
\end{align*}

For $\gamma\in \{\gamma_A, \gamma_B\}$ define for $\eta \in \{0,1\}$ $\ket{\eta_{\gamma}}:=(R^{\gamma})^*\ket{\eta}=R^{-\gamma}\ket{\eta}$, so that
\[
P_{3;A}^{\gamma_A}(\{+\})=\ket{0_{\gamma_A}}\!\bra{0_{\gamma_A}}\otimes I,
\qquad
P_{3;B}^{\gamma_B}(\{+\})=I\otimes \ket{0_{\gamma_B}}\!\bra{0_{\gamma_B}}.
\]
Using $\ket{h_{\rm Bell}}=\frac1{\sqrt2}(\ket{00}+\ket{11})$, we compute
\begin{align*}
\iprod{P_{3;A}^{\gamma_A}(\{+\})\,h_{\rm Bell}}{h_{\rm Bell}}
&=\bigl\|(\bra{0_{\gamma_A}}\otimes I)h_{\rm Bell}\bigr\|^2 \\
&=\frac12\bigl\|\braket{0_{\gamma_A}}{0}\ket{0}+\braket{0_{\gamma_A}}{1}\ket{1}\bigr\|^2 \\
&=\frac12\bigl(|\braket{0_{\gamma_A}}{0}|^2+|\braket{0_{\gamma_A}}{1}|^2\bigr)
=\frac12.  
\intertext{Also,}
\iprod{P_{3;A}^{\gamma_A}(\{-\})h_{\rm Bell}}{h_{\rm Bell}}
& = 1-\iprod{P_{3;A}^{\gamma_A}(\{+\})\,h_{\rm Bell}}{h_{\rm Bell}}=\frac12.
\intertext{The same argument (now working in the second tensor factor) gives}
\iprod{P_{3;B}^{\gamma_B}(\{\eps\})\,h_{\rm Bell}}{h_{\rm Bell}} & =\frac12,
\qquad \eps\in\{+,-\}.
\end{align*}
Combining this with the initial conjugation identities, we conclude that for all $\eps\in\{+,-\}$,
\begin{align*}
\iprod{P_{3;A}(\{\eps\})\,Uh_{\rm Bell}}{Uh_{\rm Bell}}
& =\iprod{P_{3;A}(\{\eps\})\,h_{\rm Bell}}{h_{\rm Bell}},
\\
\iprod{P_{3;B}(\{\eps\})\,Uh_{\rm Bell}}{Uh_{\rm Bell}}
& =\iprod{P_{3;B}(\{\eps\})\,h_{\rm Bell}}{h_{\rm Bell}}.
\end{align*}
\end{proof}

\begin{remark}
The conclusion of Theorem~\ref{thm:hide-2-intrinsic} holds for each of the four Bell states
\[
\frac1{\sqrt2}(\ket{00}\pm\ket{11}),
\quad
\frac1{\sqrt2}(\ket{01}\pm\ket{10}).
\]
Similarly one checks that every rotation
$R^{\gamma_A}\otimes R^{\gamma_B}$ leaves the $P_{3;A}$- and $P_{3;B}$-marginals equal to
$(\frac12,\frac12)$.
\end{remark}

\subsubsection{Reshuffling the Holevo space}
Theorem \ref{thm:Holevo-invariance-under-rotation} implies that for each fixed polariser setting $(\gamma_A,\gamma_B)$, the conjugation
$U^{\gamma_A,\gamma_B}=R^{\gamma_A}\otimes R^{\gamma_B}$ identifies the ``rotated'' algebra
$\calM^{\gamma_A,\gamma_B}$ with the unrotated algebra $\calM$.
Importantly, this identification depends on the pair of angles $(\gamma_A,\gamma_B)$. In fact, for each {\em fixed} setting $(\gamma_A,\gamma_B)$, the Holevo space
$\mathscr P(\C^2\otimes\C^2)/\!\sim_{\mathcal M^{\gamma_A,\gamma_B}}$ is canonically identified
with $\Delta^3$ by mapping $\ket{h}\in \calP(\C^2\otimes\C^2)$ to its four joint outcome probabilities
\[
\bigl(p_{++}^{\gamma_A,\gamma_B}(h),\,p_{+-}^{\gamma_A,\gamma_B}(h),\,p_{-+}^{\gamma_A,\gamma_B}(h),\,p_{--}^{\gamma_A,\gamma_B}(h)\bigr),
\]
where $p^{\gamma_A,\gamma_B}_{\eps_A\eps_B}(h)
:=\iprod{P^{\gamma_A,\gamma_B}_{3;C}(\{(\eps_A,\eps_B)\})h}{h}$ for $\eps_A,\eps_B\in\{+,-\}$.
But when we compare different settings {\em inside a single chosen reference copy} of $\Delta^3$,
we use the setting-dependent conjugation $U^{\gamma_A,\gamma_B}=R^{\gamma_A}\otimes R^{\gamma_B}$;
this comparison is therefore non-canonical, and
the corresponding equivalence relations on $\calP(\C^2\otimes\C^2)$ differ with the choice of angles: passing from $\calM$
to $\calM^{\gamma_A,\gamma_B}$ changes the equivalence classes according to the
rule
\[
 [h]_{\calM^{\gamma_A,\gamma_B}}\mapsto
[U^{\gamma_A,\gamma_B}h]_{\calM}.
\]
Thus, varying $(\gamma_A,\gamma_B)$ does not change the quotient space up to homeomorphism, but it ``repartitions'' the underlying pure state space into different indiscernibility classes.

The final theorem translates this line of thought into a precise incompatibility result.

\begin{theorem}\label{thm:incompatibility}
Let $(\gamma_A,\gamma_B)$ and $(\gamma_A',\gamma_B')$ be two pairs of rotation
angles. If
\[
\gamma_A-\gamma_A'\notin \frac{\pi}{2}\Z
\quad\text{or}\quad
\gamma_B-\gamma_B'\notin \frac{\pi}{2}\Z,
\]
then the von Neumann algebras $\calM^{\gamma_A,\gamma_B}$ and
$\calM^{\gamma_A',\gamma_B'}$ are incompatible in the following sense:
\begin{enumerate}[\rm(1)]
 \item $\calM^{\gamma_A,\gamma_B}\not\subseteq \calM^{\gamma_A',\gamma_B'}$
 and
 $\calM^{\gamma_A',\gamma_B'}\not\subseteq \calM^{\gamma_A,\gamma_B}$;
 \item there exist $O\in \calM^{\gamma_A,\gamma_B}$ and
 $O'\in \calM^{\gamma_A',\gamma_B'}$ that do not commute.
\end{enumerate}
Moreover, if $\gamma_A-\gamma_A'\in\frac{\pi}{2}\Z$ and
$\gamma_B-\gamma_B'\in\frac{\pi}{2}\Z$, then
$\calM^{\gamma_A,\gamma_B}=\calM^{\gamma_A',\gamma_B'}$.
\end{theorem}

\begin{proof}[Proof of Theorem \ref{thm:incompatibility}]
Using computational basis notation, write $P_+:=P_3(\{+\})= \ket{0}\!\bra{0}$ and $P_+^{(\gamma)}:=(R^\gamma)^*P_+R^\gamma$.
Note that $P_+^{(\gamma)}$ is the rank--one projection onto the line spanned by
$(R^\gamma)^* \ket{0}$.

A direct calculation gives
\[
P_+^{(\gamma)}=
\begin{pmatrix}
\cos^2\gamma & \sin\gamma\cos\gamma\\
\sin\gamma\cos\gamma & \sin^2\gamma
\end{pmatrix},
\]
and hence
\[
[P_+^{(\gamma_1)},P_+^{(\gamma_2)}]
=\frac12\sin\bigl(2(\gamma_2-\gamma_1)\bigr)
\begin{pmatrix} 0 & 1 \\ -1 & 0\end{pmatrix}
=-\frac12\sin\bigl(2(\gamma_1-\gamma_2)\bigr)
\begin{pmatrix} 0 & 1 \\ -1 & 0\end{pmatrix}.
\]
Hence $[P_+^{(\gamma_1)},P_+^{(\gamma_2)}]=0$ if and only if
$\gamma_1-\gamma_2\in \frac{\pi}{2}\Z$.

Now recall that $P_{3;C}^{\gamma_A,\gamma_B}$ has atoms
\[
P_{3;C}^{\gamma_A,\gamma_B}\bigl(\{(\eps_A,\eps_B)\}\bigr)
= P_{\eps_A}^{(\gamma_A)}\otimes P_{\eps_B}^{(\gamma_B)},
\qquad \eps_A,\eps_B\in\{+,-\},
\]
where \(P_{-}^{(\gamma)}:=I-P_+^{(\gamma)}\) for $\gamma\in\{\gamma_A,\gamma_B\}$.
Hence the von Neumann algebra $\calM^{\gamma_A,\gamma_B}$ generated by the range of
$P_{3;C}^{\gamma_A,\gamma_B}$ is generated by these four tensor-product projections.
In particular,
\begin{align*}
P_+^{(\gamma_A)}\otimes I
&= \bigl(P_+^{(\gamma_A)}\otimes P_+^{(\gamma_B)}\bigr)
   +\bigl(P_+^{(\gamma_A)}\otimes (I-P_+^{(\gamma_B)})\bigr)
\in \calM^{\gamma_A,\gamma_B},
\end{align*}
and similarly $P_+^{(\gamma_A')}\otimes I\in\calM^{\gamma_A',\gamma_B'}$.

Assume first that $\gamma_A-\gamma_A'\notin\frac{\pi}{2}\Z$.
Set
\[
O:=P_+^{(\gamma_A)}\otimes I\in\calM^{\gamma_A,\gamma_B},
\qquad
O':=P_+^{(\gamma_A')}\otimes I\in\calM^{\gamma_A',\gamma_B'}.
\]
Then
\[
[O,O']=[P_+^{(\gamma_A)},P_+^{(\gamma_A')}]\otimes I\neq 0
\]
by the claim, proving (2). To deduce (1), suppose for contradiction that
$\calM^{\gamma_A,\gamma_B}\subseteq \calM^{\gamma_A',\gamma_B'}$.
Then $O\in \calM^{\gamma_A',\gamma_B'}$, and of course $O'\in\calM^{\gamma_A',\gamma_B'}$.
But $\calM^{\gamma_A',\gamma_B'}$ is abelian, as it generated by the commuting range of a PVM,
so $[O,O']=0$, contradiction. Hence
$\calM^{\gamma_A,\gamma_B}\not\subseteq \calM^{\gamma_A',\gamma_B'}$.
The reverse non-inclusion follows by symmetry. The case
$\gamma_B-\gamma_B'\notin\frac{\pi}{2}\Z$ is handled in the same way using $I\otimes P_+^{(\gamma_B)}$.

Finally, if $\gamma_A-\gamma_A'\in\frac{\pi}{2}\Z$ and $\gamma_B-\gamma_B'\in\frac{\pi}{2}\Z$,
then $P_+^{(\gamma+\pi)}=P_+^{(\gamma)}$ and $P_+^{(\gamma+\pi/2)}=I-P_+^{(\gamma)}$,
so the four atoms $P_{\eps_A}^{(\gamma_A)}\otimes P_{\eps_B}^{(\gamma_B)}$ generate the same von Neumann
algebra as the four atoms $P_{\eps_A}^{(\gamma_A')}\otimes P_{\eps_B}^{(\gamma_B')}$ (up to relabelling $\eps_A,\eps_B$).
Thus $\calM^{\gamma_A,\gamma_B}=\calM^{\gamma_A',\gamma_B'}$.
\end{proof}

In particular, unless both angle differences lie in $\frac{\pi}{2}\Z$, there is no \emph{single abelian}
von Neumann algebra $\mathcal A\subseteq\calL(\C^2\otimes\C^2)$ such that
$\mathcal M^{\gamma_A,\gamma_B}\cup \mathcal M^{\gamma_A',\gamma_B'}\subseteq\mathcal A$; hence no single
commutative Holevo lift can cover both settings without adjoining the setting choice as an extra classical variable.

\section{Outlook}
In the famous experiment performed by Alain Aspect and his collaborators \cite{Aspect},
the polariser angles $\gamma_A$ and $\gamma_B$ are not fixed constants, but are randomly selected
from the sets $\{a_1,a_2\}$ and $\{b_1,b_2\}$ in each run. This can be modelled by considering
$\gamma_A$ and $\gamma_B$ as independent random variables
$\gamma_A : \Omega_{\rm pol} \to \{a_1,a_2\}$ and $\gamma_B : \Omega_{\rm pol} \to \{b_1,b_2\}$ on
a probability space for the polariser angles
$(\Omega_{\rm pol},\P_{\rm pol})$, with
\[
\P_{\rm pol}(\gamma_A=a_1)=\P_{\rm pol}(\gamma_A=a_2)=\frac12,
\qquad
\P_{\rm pol}(\gamma_B=b_1)=\P_{\rm pol}(\gamma_B=b_2)=\frac12.
\]
Conditioned on a fixed realisation $(\gamma_A,\gamma_B)=(a_i,b_j)$, one obtains one of the four
``sub-experiments'' analysed above; for each such fixed setting our results provide a classical
(Holevo) description of the corresponding measurement statistics.

Theorem~\ref{thm:incompatibility} shows, however, that different settings typically correspond to
\emph{incompatible} observable algebras. From Charlie's perspective -- who keeps track of which
setting was used in each run -- this obstructs a description of the \emph{combined} four-setting
experiment within a single Holevo space without adjoining the setting choice as an extra
classical variable. While each fixed setting admits a Holevo space,
there is in general no canonical way to ``glue'' these descriptions into one `global Holevo space'
that is valid uniformly across all settings.

From the local viewpoints the situation is complementary. Theorem~\ref{thm:hide-2-intrinsic} shows
that the local marginal statistics accessible to Alice and Bob do not reveal the chosen pair
$(\gamma_A,\gamma_B)$. In particular, without additional information about the settings, the
two local perspectives cannot simply be merged into one global account of the experiment.

These two observations motivate the companion paper \cite{NeeWaa}. This paper
treats the setting choices themselves as part of the probabilistic description and constructs an
explicit classical probability space for the full four-setting protocol (including both the
random variables $\gamma_A,\gamma_B$ and the measurement outcomes) which reproduces the conditional
distributions of the four sub-experiments. In this sense \cite{NeeWaa} is a continuation of
the present paper: it extends the fixed-setting analysis to a single classical model for the randomised multiple-run Bell experiment.

At the same time, \cite{NeeWaa} differs in spirit from the present paper. Whereas the present paper is primarily a
structural study of indiscernibility, the focus of \cite{NeeWaa} is interpretational: the global construction is used as a
tool to reassess which locality/causality assumptions are doing the work in Bell-type arguments, and to discuss how such assumptions are (and are not) supported by the usual relativistic
motivations.

\bigskip
{\em Acknowledgment} -- We thank one of the referees for pointing out some mathematical oversights in the original submission. We used several versions of GPT as a brainstorming partner helping us to uncover proof strategies for some of the theorems and to improve the presentation of the paper.

\bibliographystyle{plain}
\bibliography{indiscernible}

\end{document}